\documentclass[12pt]{article}
\usepackage{amsfonts}
\usepackage{amssymb}
\usepackage{graphicx}
\usepackage{amsmath}
\usepackage{amstext}
\usepackage{geometry}
\usepackage[onehalfspacing]{setspace}
\usepackage{xcolor}
\usepackage{harvard}

\def\limfunc#1{\mathop{\rm #1}}%
\def\QTR#1#2{{\csname#1\endcsname {#2}}}%

\newtheorem{corollary}{Corollary}

\newtheorem{definition}{Definition}
\newtheorem{example}{Example}
\newtheorem{example2}{Example}

\newtheorem{lemma}{Lemma}

\newtheorem{proposition}{Proposition}

\newenvironment{proof}[1][Proof]
{\textbf{#1.} }{\ \rule{0.5em}{0.5em}}
\renewcommand{\cite}{\citeasnoun}
\begin{document}

\title{The Economics of Social Data\thanks{%
Bergemann and Bonatti acknowledge financial support through NSF\ Grant
SES-1948692. Bergemann and Gan acknowledge financial support from the
Omidyar and Knight Foundation. We thank Joseph Abadi, Daron Acemo\u{g}lu,
Susan Athey, Steve Berry, Nima Haghpanah, Nicole Immorlica, Al Klevorick,
Scott Kominers, Annie Liang, Roger McNamee, Jeanine Mikl\'{o}s-Thal, Enrico
Moretti, Stephen Morris, Denis Nekipelov, Asu \"{O}zda\u{g}lar, Fiona
Scott-Morton, Shoshana Vasserman, Glen Weyl, and Kai-Hao Yang for helpful
discussions. We also thank Michelle Fang and Miho Hong for valuable research
assistance and the audiences at numerous seminars and conferences for their
productive comments.}}
\author{Dirk Bergemann\thanks{
Department of Economics, Yale University, New Haven, CT 06511, \texttt{\
dirk.bergemann@yale.edu}.} \and Alessandro Bonatti\thanks{%
MIT\ Sloan School of Management, Cambridge, MA 02142, \texttt{bonatti@mit.edu%
}.} \and Tan Gan\thanks{
Department of Economics, Yale University, New Haven, CT 06511, \texttt{\
tan.gan@yale.edu}.}}
\date{\today }
\maketitle

\begin{abstract}
A data intermediary acquires signals from individual consumers regarding
their preferences. The intermediary resells the information in a product
market wherein firms and consumers tailor their choices to the demand data.
The social dimension of the individual data---whereby a consumer's data are
predictive of others' behavior---generates a data externality that can
reduce the intermediary's cost of acquiring the information. The
intermediary optimally preserves the privacy of consumers' identities if and
only if doing so increases social surplus. This policy enables the
intermediary to capture the total value of the information as the number of
consumers becomes large.\medskip \medskip

\medskip

\newpage
\end{abstract}

\newpage

\section{Introduction\label{intro}}

\paragraph{Individual and Social Data}

The rise of large digital platforms---such as Facebook, Google, and Amazon
in the US and JD, Tencent and Alibaba in China---has led to the
unprecedented collection and commercial use of individual data. The steadily
increasing user bases of these platforms generate massive amounts of data
about individual consumers, including their preferences, locations, friends,
and political views. In turn, many of the services provided by large
Internet platforms rely critically on these data. The availability of
individual-level data permits refined search results, personalized product
recommendations, informative ratings, timely traffic data, and targeted
advertisements.

The recent disclosures on the use and misuse of social data by digital
platforms have prompted regulators to limit the largely unsupervised use of
individual data by these companies. As a result, nearly all proposed and
enacted regulation to date aims to ensure that consumers retain control over
their data. However, the \emph{digital privacy paradox} (e.g., \cite{atct17}%
) indicates that even small monetary incentives may lead individuals to
relinquish their private data. The low cost of acquiring private
data--seemingly at odds with consumers' stated preferences over their
privacy-- drives the appetite of platforms to gather information and may
undermine the efficacy of regulation.\footnote{%
The recent report by \cite{furm19} identifies \textquotedblleft the central
importance of data as a driver of concentration and a barrier to competition
in digital markets\textquotedblright ---a theme echoed in the reports by 
\cite{crem19} and by the \cite{stig19}.}

This article suggests a unified explanation for the digital privacy paradox
and the selective use of data for price and product choices. A key
observation is that \emph{individual data} are actually \emph{social data}:
data captured from an individual user are informative not only about that
user but also about other users with similar characteristics or behaviors.
The social dimension of the data generates a \emph{data externality}, the
sign and magnitude of which depend on the structure of the data and on the
use of the gained information.

Many digital platforms use social data to increase the value provided by
their services. Google search results, for example, are informed by the
choices of previous users. Indeed, the search engine is only successful
because it mediates information among many consumers. Likewise, Amazon uses
data collected from other consumers to curate a \textquotedblleft
recommended for you\textquotedblright\ list of products and explicitly
suggests goods that are \textquotedblleft frequently bought
together.\textquotedblright\ Finally, YouTube and Waze tailor their
suggestions of videos and traffic directions to each user's preferences, as
estimated by combining network data and individual data. However, social
data also facilitates surplus extraction, e.g., individual shopping data
convey information about the willingness to pay of consumers with similar
purchase histories.

We analyze three critical aspects of the economics of social data. First, we
consider how the collection and transmission of individual data change the
terms of trade among consumers, firms (e.g., advertisers), and data
intermediaries (e.g., large Internet platforms that sell targeted
advertising space). Second, we examine how the social dimension of the data
magnifies the value of individual data for platforms and facilitates the
acquisition of large datasets. Third, we analyze how data intermediaries
with market power manipulate the trade-offs induced by social data through
the aggregation and the precision of the information that they provide about
consumers.

\paragraph{A Model of Data Intermediation}

We develop a framework to evaluate the flow and allocation of individual
data in the presence of data externalities. Our model focuses on three types
of economic agents: consumers, firms, and data intermediaries. These agents
interact in two distinct but linked markets: a \emph{data market} and a 
\emph{product market}.

In the product market, each consumer (she) determines the quantity that she
wishes to purchase, and a single producer (he) sets the unit price at which
he offers a product to the consumers. Initially, each consumer has private
information about her willingness to pay for the firm's product. This
information consists of a signal with two additive components: a \emph{\
fundamental} component and a \emph{noise} component. The fundamental
component represents her willingness to pay, and the noise component
reflects that her initial information might be imperfect. Both components
can be correlated across consumers: in practice, different consumers'
preferences can exhibit common traits, and consumers might undergo similar
experiences that induce correlation in their errors.

In the data market, a monopolist intermediary acquires demand information
from the individual consumers in exchange for a monetary payment. The
intermediary then chooses how much information to share with the other
consumers and how much information to sell to the producer. Sharing data
with consumers allows them to tailor their demand to their true preferences.
Sharing data with the producer enables more tailored and possibly
personalized prices.

We introduce a rich model for the structure of individual data, which allows
for correlation in the fundamentals as well as in the noise terms across the
individuals. We view this richness in the data structure as the defining
element of data in digital platforms. In contrast, we adopt a more specific
model for product market interaction whereby a monopolist seller charges
linear prices for variable quantity. However, our main insights apply to any
product market where (a) data sharing teaches consumers about their
preferences, and (b) a firm seeks to extract the consumers' surplus.%
\footnote{%
In fact, the value of social data in Section \ref{dw} can be computed for
alternative specification of the product market and a general result for
anonymization is obtained in Section \ref{opdi}. Finally in Section \ref%
{limits}, we consider a richer environment where a merchant offers different
varieties to consumers with heterogeneous tastes for his products. In that
case, the firm's actions both generate value and attempt to capture it.}

\paragraph{The Value of Social Data\label{vsd}}

Collecting data from multiple consumers helps any market participant to
filter the individual signals. This process can occur through two channels.
First, in a market where the noise terms are largely idiosyncratic, a large
sample size filters out errors and identifies any common fundamentals.
Second, in a market with largely idiosyncratic fundamentals, many
observations filter out demand shocks and identify common noise terms,
thereby estimating individual fundamentals by differencing (Proposition \ref%
{impact}).

However, the choice by each consumer to share her data with the intermediary
is guided only by her private benefits and costs, not by the information
gains she generates with her actions. Thus, the intermediary must compensate
each individual consumer only to the extent that the disclosed information
affects her own welfare \emph{on the margin}. Critically, the platform does
not have to compensate the individual consumer for any changes in her
welfare caused by the information deduced from other consumers' signals.

Therefore, social data drive a wedge between the socially efficient and
profitable uses of information. First, the cost of acquiring individual data
can be substantially less than the value of the information to the platform.
Second, although many uses of consumer information exhibit positive
externalities, very little prevents the platform from trading data for
profitable uses that are in fact harmful to consumers (Proposition \ref{prof}%
).\footnote{%
Recent empirical work on the effects of privacy regulation such as the
European Union's General Data Protection Regulation (e.g., \cite{archsa20}
and \cite{joshgo20}), indicates that data externalities are relevant for
consumers' and businesses' decisions to share their data. In the United
States, legislators are also increasingly aware of the consequences of data
externalities. In particular, the US\ House \cite{house20} reports that
\textquotedblleft \lbrack ...] the social data gathered through [a
platform's] services may exceed their economic value to
consumers.\textquotedblright}

More generally, the data externality can induce too much or too little trade
in data. Indeed, {the condition for data intermediation to yield positive
profits is qualitatively different from the condition for intermediation to
yield social welfare gains. In particular, the intermediary obtains positive
profits when a large number of consumers exhibit a strong degree of
correlation in their preferences: this allows the intermediary to acquire
the consumers' data in exchange for minimal compensation. On the other hand,
welfare improvements depend on how much additional information each consumer
can obtain about her own preferences from the other consumers' signals.
Therefore, when many consumers have strongly correlated preferences and very
precise signals, data sharing is detrimental to welfare but profitable to
the intermediary. Conversely, there are data structures (e.g., ones with
independent fundamentals and strongly correlated error terms) for which data
sharing is beneficial to consumers but unprofitable for the data
intermediary.}

\paragraph{Equilibrium Data-Sharing Policies}

A natural question is then whether the data market imposes any limitations
at all on equilibrium information sharing. To shed light on this issue, we
consider the choice of whether to reveal the consumers' identities to the
producer or to collect anonymous data. When consumers are homogeneous ex
ante, we show that the intermediary collects anonymous data if and only if
the transmission of identity data reduces total surplus. Therefore, even if
the data \emph{transmission }may be socially detrimental, the optimal choice
of data \emph{anonymization }is socially efficient (Proposition \ref{si}).

In our model of linear price discrimination, collecting anonymous data
amounts to selling aggregate, market-level information to the producer. With
this choice, the intermediary does not enable the producer to set
personalized prices: the data are transmitted but disconnected from the
users' personal profiles. In other words, the role of social data provides a
more nuanced ability to determine the modality of information acquisition
and use.

Under anonymized data intermediation, the gap between the social value of
the data and the price of the data widens when the number of consumers
increases. In particular, as the sources of data multiply, the contribution
of each individual consumer to the aggregate information shrinks, which
drives down the individual payments to consumers, and possibly the total
payment as well (Proposition \ref{lm}).

We develop a general anonymization result (Proposition \ref{ga})\ and extend
the model in several directions. We first introduce consumer heterogeneity
by considering multiple groups of consumers. We find that the intermediary
aggregates the data at least to the level of the coarsest partition of
homogeneous consumers, although further aggregation is profitable when the
number of consumers is small. The resulting group pricing---discriminatory
pricing based on observable characteristics, such as location---has welfare
consequences between those of complete privacy and those of price
personalization (Proposition \ref{gs}).

We then consider a model in which the producer can choose prices and product
characteristics to match an additional horizontal (taste) dimension of the
consumers' preferences. The resulting data policy then aggregates the
vertical dimension but not the horizontal dimension, thereby enabling the
producer to offer personalized product recommendations but not personalized
prices (Proposition \ref{vag3}). Finally, in the Online Appendix, we explore
the data intermediary's ability to offer privacy guarantees by collecting
less than perfect information about the consumers' signals.

\paragraph{Implications and Applications}

There exist a wide variety of data intermediaries and associated business
models. The point of our article is not to try to match any specific
business model but rather to speak to the general principles that seem to be
in effect in many markets. In particular, our model assumes that each
consumer is compensated directly with a monetary transfer for her individual
data. There exist a few concrete examples of such transactions (e.g.,
Nielsen offers monetary rewards to consumers for access to their browsing
and purchasing data). However, most data intermediaries compensate their
users via the quality of the services they offer, e.g., social networks,
search, mail, video. These services are nominally free, but in practice
(through more or less transparent terms and conditions) they are fueled by
the data generated by their users.

Likewise, digital platforms occasionally transfer consumers' data to
merchants for a fee. Much more often, however, they monetize their data by
selling access to targeted advertising campaigns---see the report by the 
\cite{cma20}. This enables the merchants to reap the value of information,
by conditioning their messages and their prices on the consumers'
preferences, without directly observing their data. These transactions
amount to indirect sales of information, as discussed in \cite{bebo19}.

Our model's predictions for the nature of the externality, for the
equilibrium allocation of data, and for its welfare properties then depend
on how targeted advertising affects total surplus in any given market.
Specifically, if advertising generates value by matching of buyers and
sellers, as in \cite{bebo11}, our model would predict the intermediation of
individual-level information through very detailed targeting categories. If,
instead, targeted advertisements facilitate inefficient price
discrimination, the model would predict coarser targeting categories that
induce market-level or group-level pricing.

\paragraph{Related Literature}

This article contributes to the growing literature on data markets recently
surveyed in \cite{bebo19}. In particular, the role of data externalities in
the socially excessive diffusion of personal data has been a central concern
in \cite{chjk19} and \cite{ammo21}, both considering a model with many
buyers and one firm that acts as an integrated data intermediary and seller.

\cite{chjk19} introduce information externalities into a model of monopoly
pricing with unit demand. Each consumer is described by two independent
random variables: her willingness to pay for the monopolist's service and
her sensitivity to a loss of privacy. The purchase of the service by the
consumer requires the transmission of personal data. From the collected
data, the seller gains additional revenue, depending on the proportion of
units sold and the volume of data collected. The total nuisance cost paid by
each consumer depends on the total number of consumers sharing their
personal data. Thus, the optimal pricing policy of the monopolist yields
excessive loss of privacy, relative to the social welfare maximizing policy.

\cite{ammo21} also analyze data acquisition in the presence of information
externalities. As in \cite{chjk19}, they consider a model with many
consumers and a single data-acquiring firm. Like the current analysis, \cite%
{ammo21} propose an explicit statistical model for their data; the model
allows the authors to assess the loss of privacy for the consumer and the
gains in prediction accuracy for the firm. Their analysis then pursues a
different, and largely complementary, direction from ours. In particular,
they analyze how consumers with heterogeneous privacy concerns trade
information with a data platform and derive conditions under which the
equilibrium allocation of information is (in)efficient.

In contrast to these two important contributions, we explicitly introduce a
data intermediary with objectives distinct from either the consumers or the
seller. We consider rich data structures for fundamentals and noise terms
that capture the wide range of social data on digital platforms. This allows
us to endogenize privacy concerns and to quantify the downstream welfare
impact of data intermediation. In addition, we can investigate when and how
privacy can be partially or fully preserved through aggregation,
anonymization, and noise. Thus, we augment the analysis in the above
contributions with additional insights regarding data flows and data
intermediation.\ In particular, we show that the data externality can have a
positive or a negative impact on consumer and social surplus depending on
the data structure. In consequence, a monopolist intermediary can induce
either too little or too much information sharing in equilibrium.

\cite{actw16} survey the literature on the economics of privacy in great
detail. An early and influential article is \cite{tayl04}, who analyzes the
sales of consumer purchase histories without data externalities. More
recently, \cite{clpr16} investigate how privacy policies affect user and
advertiser behavior in a simple model of targeted advertising. The low level
of compensation that users command for their personal data is discussed in 
\cite{agjl18}, who propose sources of countervailing market power, and in a
growing body of empirical work. In particular, \cite{tang19} uses
large-scale field experiments to estimate the value that online borrowers
assign to several pieces of personal data. \cite{lin19} separates intrinsic
and instrumental preferences for privacy through a lab experiment that also
uncovers data externalities.

\cite{fagm20} provide a digital privacy model in which data collection
improves the service provided to consumers. However, as the collected data
can also leak to third parties and thus impose privacy costs, an optimal
digital privacy policy must be established. Similarly, \cite{julr20} analyze
the equilibrium privacy policy of websites that monetize information
collected from users by charging third parties for targeted access. \cite%
{grad17} considers a network game in which the level of beneficial
information sharing among the players is limited by the possibility of
leakage and a decrease in informational interdependence. \cite{allv19} study
a model of personalized pricing with disclosure by an informed consumer, and
they analyze how different disclosure policies affect consumer surplus. \cite%
{ichi20aer} studies both personalized pricing and product recommendations,
and shows that a seller benefits from committing not to use the consumer's
information to set prices. Our result on optimal anonymization and
market-level pricing has similar implications, but is entirely driven by the
data externality that appears when multiple consumers are present.

Finally, \cite{lima20} investigate how data policies can provide incentives
in principal-agent relationships. They emphasize the structure of individual
data and how the substitutes or complements nature of individual signals
determines the impact of data on incentives. \cite{ichi20} considers a
single data intermediary and asks how complements or substitutes consumer
signals affect the equilibrium price of the individual data.

\section{Model\label{mod}}

We consider an idealized trading environment with many consumers, a single
intermediary in the data market, and a single producer in the product market.

\subsection{Product Market\label{pp}}

There are finitely many consumers, labeled $i=1,...,N$. In the product
market, each consumer (she) chooses a quantity level $q_{i}$ to maximize her
net utility given a unit price $p_{i}$ offered by the producer (he):%
\begin{equation*}
u_{i}\left( w_{i},q_{i},p_{i}\right) \triangleq w_{i}q_{i}-p_{i}q_{i}-\frac{1%
}{2}q_{i}^{2}.
\end{equation*}%
Each consumer $i$ has a baseline willingness to pay for the product $%
w_{i}\in \mathbb{R}$.

The producer sets the unit price $p_{i}$ at which he offers his product to
each consumer $i$. The producer has a linear production cost 
\begin{equation*}
c\left( q\right) \triangleq c\cdot q,\ \text{for some}\ \ c\geq 0.
\end{equation*}%
The producer's profits are given by%
\begin{equation*}
\pi \left( p_{i},q_{i}\right) \triangleq \sum_{i}\left( p_{i}-c\right) q_{i}.
\end{equation*}

\subsection{Data Environment\label{da}}

The vector of willingness to pay, $w=\left( ...,w_{i},...\right) \in \mathbb{%
R}^{N}$, is distributed according to a joint distribution $F_{w}$:%
\begin{equation}
w\sim F_{w}.  \label{dv}
\end{equation}%
Initially, each consumer may have only imperfect information about her
willingness to pay. In particular, consumer $i$ observes a signal%
\begin{equation}
s_{i}\triangleq w_{i}+\sigma \cdot e_{i}\text{,}  \label{no}
\end{equation}%
where $\sigma >0$ and $e_{i}$ is consumer $i$'s error term. The error terms $%
e=\left( ...,e_{i},...\right) \in \mathbb{R}^{N}$ are independent of the
willingness to pay $w$, and they are distributed according to a joint
distribution $F_{e}$: 
\begin{equation}
e\sim F_{e}.  \label{de}
\end{equation}%
We denote by $S$ the information structure generated by the complete vector
of consumer signals $s=\left( ...,s_{i},...\right) \in \mathbb{R}^{N}$. We
allow for arbitrary distributions of fundamentals $w$ and errors $e$, and
hence arbitrary correlation structures across consumers, under the
restriction that the $\left( F_{w},F_{e}\right) $ are symmetric across
individuals. We view the richness in the data structure as represented by (%
\ref{dv}) and (\ref{de}) as the defining feature of social data in the
digital economy. In particular, the noise in the signal $s_{i}$\ of each
individual given by (\ref{no}) reflects the importance of social learning as
enabled by recommenders, rating and search engines.

Without loss of generality we assume that (i) each individual willingness to
pay $w_{i}$ has mean $\mu $ and variance $1$; (ii) individual errors $e_{i}$
have mean $0$ and variance $1$ (which is scaled by the parameter $\sigma $).

The producer knows the structure of demand and thus the common prior
distribution of consumers' willingness to pay. However, absent any
additional information, the producer does not know the realized willingness
to pay $w_{i}$ of any consumer (nor her signal $s_{i}$) prior to setting
prices. This data environment has two important features. First, any demand
information beyond the common prior comes from the signals of the individual
consumers. Second, with any amount of noise in the signals (i.e., if $\sigma
>0$), each consumer can learn more about her own demand from the signals of
the other consumers.

The following leading examples illustrate two ways in which data sharing can
help each consumer learn more about her individual willingness to pay. In,
Example \ref{e1} a new product has a common value that consumers are
imperfectly informed about.

\begin{example}[Common Preferences]
\label{e1}\quad \newline
Fundamentals $w_{i}$ are perfectly correlated and errors $e_{i}$ are
independent: $s_{i}=w+\sigma \cdot e_{i}$.
\end{example}

In this case, data sharing helps to filter out the idiosyncratic error
terms: as $N$ becomes large, the average signal across all consumers
identifies the common willingness to pay.

In Example \ref{e2}, individual consumers have independent values for a new
therapy but are exposed to a common health shock.

\begin{example}[Common Experience]
\label{e2}\quad \newline
Errors $e_{i}$ are perfectly correlated, and fundamentals $w_{i}$ are
independent: $s_{i}=w_{i}+\sigma \cdot e.$
\end{example}

Under this structure, the average signal identifies the common error
component $e$ as $N$ becomes large. All market participants can then
precisely estimate each $w_{i}$ from the difference between individual and
average signal.

As we shall see, information sharing enables learning in both examples.
However, the actions of consumers\ $-i\ $impact the surplus of consumer $i$
quite differently in the two cases, which has implications for the
equilibrium price of data. More generally, the data structure will determine
how to separate the individual and the aggregate information.

\subsection{Data Market\label{dm}}

The data market is run by a single data intermediary (it). As a monopolist
market maker, the data intermediary decides how to collect the available
information $\left( s_{i}\right) $ from each consumer and how to share it
with the other consumers and the producer. Thus, the data intermediary faces
both an information design problem and an information pricing problem.

We consider bilateral contracts between the individual consumers and the
intermediary and between the producer and the intermediary. The data
intermediary offers these bilateral contracts \emph{ex ante}, i.e., before
the realization of any demand shocks. Each bilateral contract defines a 
\emph{data policy }and\emph{\ }a \emph{data price}.

The data contract with consumer $i$ specifies a \emph{data inflow }policy $%
X_{i}$ and a fee $m_{i}\in \mathbb{R}$ paid to the consumer. The data inflow
policy describes how each signal $s_{i}\ $enters the database of the
intermediary. We restrict attention to the following two policies: (i) the 
\emph{complete (identity-revealing) }data policy $X=S$, where the
intermediary collects each consumer's signal $s_{i}$; and (ii) the \emph{%
anonymized }data policy $X=A$,\ where the intermediary collects individual
signals without identifying information. We model the anonymized data policy
as 
\begin{equation*}
A:S\rightarrow \delta \left( S\right) \text{,}
\end{equation*}%
for a random permutation of the consumers' indices $i\rightarrow \delta
\left( i\right) $. In both cases, as discussed in Section \ref{dis} below,
there are no further incentive constraints, i.e., consumers transmit their
information truthfully.

In our product market model, where the consumer's demand is linear in her
signal, the anonymized data policy $A$ is equivalent to an \emph{aggregate }%
data policy that conveys information about the average willingness to pay.
Intuitively, the anonymized data policy prevents the producer from matching
signals to consumers, i.e., from profitably charging personalized prices. In
Section \ref{idd}, we enrich the intermediary's strategy space by allowing
for data policies that collect partial information about the consumers'
signals.

A data contract with the producer specifies a \emph{data outflow }policy $Y$
and a fee $m_{0}\in \mathbb{R}$ paid by the producer. The data outflow
policy determines how each consumer's collected signal is transmitted to the
producer and to other consumers. In particular, letting $X$ denote the
intermediary's \emph{realized }data inflow, a data outflow policy $Y=\left(
Y_{0},Y_{1},\dots Y_{N}\right) $ describes how the collected data are
released to the seller,%
\begin{equation*}
Y_{0}:X\rightarrow \Delta (\mathbb{R}^{N}),
\end{equation*}%
and to each consumer, 
\begin{equation*}
Y_{i}:X\rightarrow \Delta (\mathbb{R}^{N}).
\end{equation*}%
Sharing data with other consumers is a critical design element because doing
so allows each consumer to adjust her quantity demanded at any price.
Therefore, the information received by consumers also impacts the \emph{\
producer's }willingness to pay for the intermediary's data.

The data intermediary maximizes the net revenue%
\begin{equation}
R\triangleq m_{0}-\sum\nolimits_{i=1}^{N}m_{i}.  \label{obj}
\end{equation}

\subsection{Equilibrium and Timing\label{et}}

The game proceeds sequentially. First, the terms of trade on the data market
and then the terms of trade on the product market are established. The
timing of the game is as follows:

\begin{enumerate}
\item The data intermediary offers a data inflow policy $\left(
m_{i},X_{i}\right) $\ to each consumer $i$. Consumers simultaneously accept
or reject the intermediary's offer.

\item The data intermediary offers a data outflow policy $\left(
m_{0},Y\right) $\ to the producer, based on the (known) number of consumers
who have accepted. The producer accepts or rejects the offer.

\item Consumers' signals $s$ are realized,\ and the information flows $%
\left( x,y\right) $\ are transmitted according to the terms of the data
policies.

\item The producer sets a unit price $p_{i}$\ for each consumer $i$ who
makes a purchase decision $q_{i}$, given her available information about $%
w_{i}$.
\end{enumerate}

We analyze the Perfect Bayesian Equilibria of the game. Under the timing
described above, information is imperfect but symmetric at the contracting
stage. Furthermore, when the consumer receives the intermediary's offer, she
must anticipate the intermediary's choice of data outflow policy, which
determines what data are shared with her, as well as with the producer. We
denote by $a_{0},a_{i}\in \left\{ 0,1\right\} $ the participation decisions
by the producer and by consumer $i$, respectively. A Perfect Bayesian
Equilibrium is then a tuple of inflow and outflow data policies, data and
product pricing policies, and participation decisions:%
\begin{equation*}
\left\{ \left( X^{\ast },Y^{\ast },m^{\ast }\right) ;~p^{\ast },q^{\ast
};~a^{\ast }\right\} ,
\end{equation*}%
where 
\begin{equation*}
a_{0}^{\ast }:X\times Y\times \mathbb{R\rightarrow }\left\{ 0,1\right\}
,~a_{i}^{\ast }:X_{i}\times \mathbb{R\rightarrow }\left\{ 0,1\right\} ,
\end{equation*}%
such that (i) the producer maximizes his expected profits, (ii) the
intermediary maximizes its expected revenue, and (iii) each consumer
maximizes her net utility. In our baseline analysis, we focus on the best
equilibrium for the data intermediary; in the best equilibrium, every
consumer accepts the offer from the data intermediary. We discuss a unique
implementation in Section \ref{dac}.

Figure \ref{diagram1} summarizes the information and value flow in the data
and product markets.

\begin{figure}[htbp]
	\label{diagram1}%
	\centering
	\caption{Data and Value Flows}%
	\includegraphics[width=0.6\textwidth]{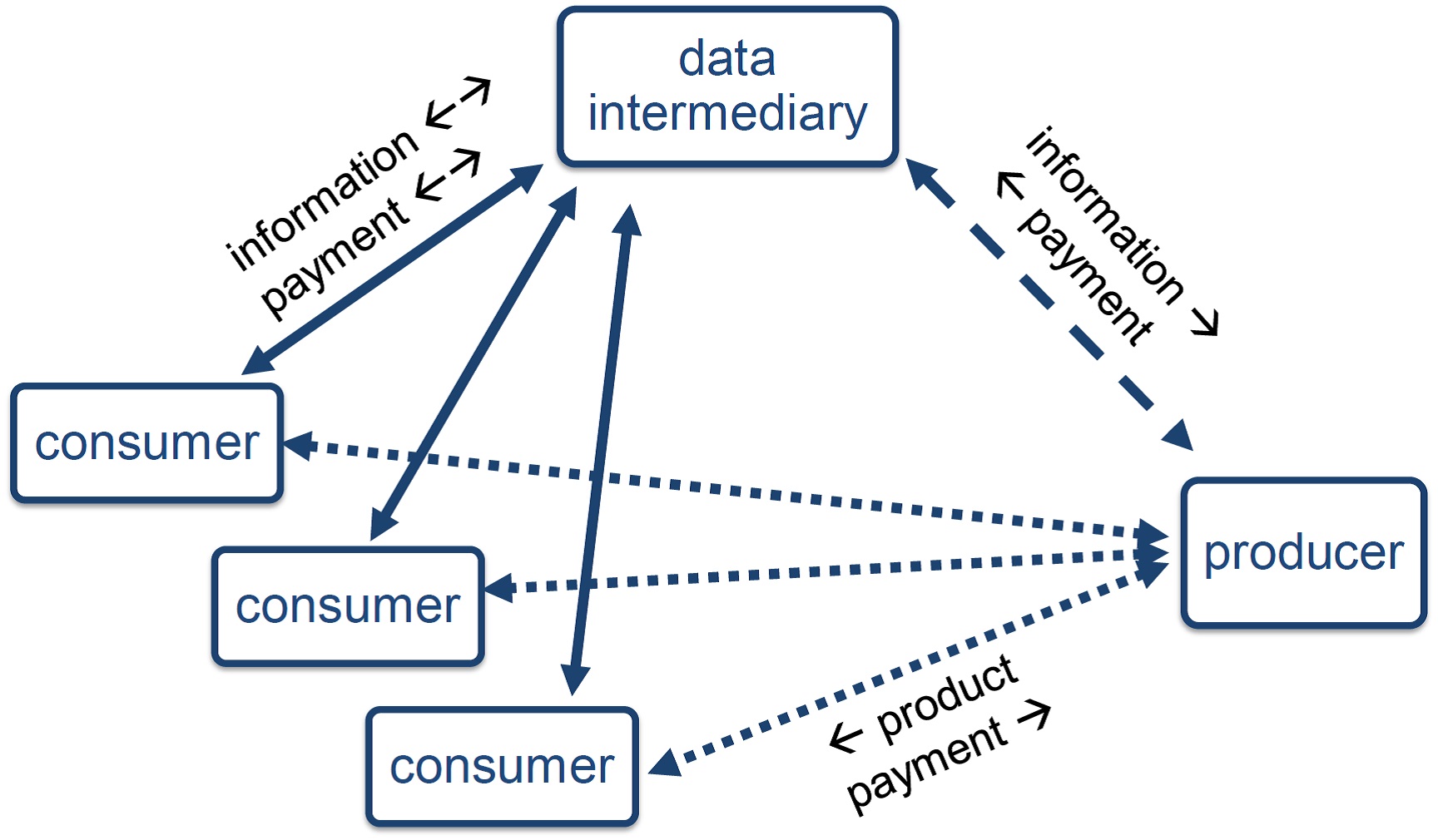}
\end{figure}

%EndExpansion

\subsection*{Discussion of Model Features\label{dis}}

\noindent \textbf{Participation Constraints\qquad }The participation
constraints of every consumer and of the producer are required to hold at
the \emph{ex ante} level. Thus, the consumers agree to the data policy
before the realization of their signals. The choice of \emph{ex ante}
participation constraints captures the prevailing conditions under which
users interact with large digital platforms. For instance, users of services
on Amazon, Facebook, or Google typically establish an account and accept the
\textquotedblleft terms of service\textquotedblright\ before making any
specific query or post. Through the lens of our model, the consumer requires
a level of compensation that allows her to profitably share the information
in expectation. Upon agreeing to participate, there are no further incentive
compatibility constraints on the transmission of her information.\smallskip

\noindent \textbf{Lack of Commitment\qquad }As we mentioned above, targeted
advertising is the primary source of revenue for digital platforms.
Consequently, the data intermediary in our sequential game sells the
consumers' data only to the producer, and cannot commit to withholding any
information from him. Similarly, the intermediary's choice of data outflow
policy occurs after consumers are enlisted but before their data are
realized. This assumption captures the limited ability of a platform to
write advertising contracts contingent on, say, the volume of activity
taking place at any point in time. \smallskip

\noindent \textbf{Linear Price Discrimination\qquad }The producer in our
baseline model uses the consumers' data to charge a unit price to each
consumer. Whether richer pricing instruments enable online merchants to
implement more sophisticated forms of price discrimination is largely an
empirical question. However, as we discuss in Section \ref{gasec}, the
general anonymization result of Proposition \ref{ga} guides the
intermediary's optimal data policy even if the producer can extract more of
the surplus generated through better information. Indeed, even in the case
of perfect price discrimination, the presence of the data externality would
continue to drive a gap between equilibrium and socially efficient
allocation.

\section{Value of Social Data\label{dw}}

\subsection{Data Sharing and Product Market\label{dspm}}

In our extensive form game, given the realized data inflow, the intermediary
offers a data outflow policy to the producer. This policy specifies both the
fee $m_{0}$ and the flow of information to all market participants,
including the consumers. The data outflow policy thus determines both how
well informed the producer is and how well informed his customers are.

Because the information is being sold to the producer, the intermediary will
choose a data outflow policy that maximizes the producer's profits, and then
extracts the value of this information through the fee $m_{0}$. The
intermediary's choice of outflow policy is simplified considerably by the
following result, which allows us to restrict attention to policies in which
every consumer is weakly better informed than the producer.

\begin{lemma}[Data Outflow Policy]
\label{tan}\quad \newline
For any data inflow $X$, it is without loss of generality to consider data
outflow policies where $Y_{0}(X)$ is measurable with respect to each $%
Y_{i}(X)$.
\end{lemma}

Interestingly, the result of Lemma \ref{tan} does not rely on the specific
nature of the product market interaction. The key observation is that, if
the producer were better informed, the prices charged would convey
information to consumers about their own willingness to pay. The ensuing
signaling incentives impose a cost on the producer because he will need to
deviate from the prices that maximize his profits, holding fixed the
consumers' beliefs. The intermediary can then increase the producer's
profits (weakly) by revealing any information contained in the equilibrium
prices directly to the consumers. Furthermore, this improvement is strict
if, when the consumers receive this additional information, the producer
modifies his choice of price.

Under the data outflow policies of Lemma \ref{tan}, we can easily quantify
the value of information for consumers and producers. The shared data help
each consumer estimate her own willingness to pay. For the producer, the
shared data enable a more informed pricing policy. In particular, given the
realized data outflow $y\in Y$, the optimal pricing policy for the producer
consists of a vector of (potentially personalized) prices $p_{i}^{\ast }$,
thus resulting in a vector of individual quantities $q_{i}^{\ast }$\
purchased. We denote the predicted value of consumer $i$'s willingness to
pay, given the signals $\left( s_{i},y_{i}\right) $ by:%
\begin{equation*}
\widehat{w}_{i}\left( s_{i},y_{i}\right) \triangleq \mathbb{E}\left[
w_{i}\mid s_{i},y_{i}\right] .
\end{equation*}%
The realized demand of consumer $i$ is given by%
\begin{equation*}
q_{i}\left( s_{i},y_{i},p\right) =\widehat{w}_{i}\left( s_{i},y_{i}\right)
-p.
\end{equation*}%
As $Y_{0}$ is (weakly) less informative than $Y_{i}$, the producer chooses
the optimal price%
\begin{equation*}
p_{i}^{\ast }\left( y_{0}\right) =\frac{\mathbb{E}[\widehat{w}_{i}\left(
s_{i},y_{i}\right) |Y_{0}]+c}{2}=\frac{\widehat{w}_{i}\left( y_{0}\right) +c%
}{2}\text{,}
\end{equation*}%
which results in the equilibrium quantity: 
\begin{equation*}
q_{i}^{\ast }\left( s_{i},y_{i},y_{0}\right) \triangleq q_{i}\left(
s_{i},y_{i},p_{i}^{\ast }\left( y_{0}\right) \right) .
\end{equation*}

The ex ante expected profit of the producer from interacting with consumer $%
i $ is given by%
\begin{equation*}
\Pi _{i}\left( \left( S_{i},Y_{i}\right) ,Y_{0}\right) \triangleq \mathbb{E}%
\left[ \pi \left( p_{i}^{\ast }\left( y_{0}\right) ,q_{i}^{\ast }\left(
s_{i},y_{i},y_{0}\right) \right) \left\vert Y_{0}\right. \right] =\frac{1}{4}%
\mathbb{E}\left[ (\widehat{w}_{i}\left( y_{0}\right) -c)^{2}\left\vert
Y_{0}\right. \right] .
\end{equation*}%
The first argument in $\Pi _{i}\left( \cdot ,\cdot \right) $ refers to
consumer $i$'s information structure $\left( S_{i},Y_{i}\right) $ and the
second argument refers to the producer's information structure $Y_{0}$.
Similarly, we denote the indirect utility of consumer $i$ as%
\begin{equation*}
U_{i}\left( \left( S_{i},Y_{i}\right) ,Y_{0}\right) \triangleq \mathbb{E}%
\left[ u_{i}\left( w_{i},q_{i}^{\ast }\left( s_{i},y_{i},y_{0}\right)
,p_{i}^{\ast }\left( y_{0}\right) \right) \left\vert S_{i},Y_{i}\right. %
\right] =\frac{1}{2}\mathbb{E}\left[ (\widehat{w}_{i}\left(
s_{i},y_{i}\right) -p_{i}^{\ast }\left( y_{0}\right) )^{2}\left\vert
S_{i},Y_{i}\right. \right] .
\end{equation*}

\subsection{Welfare Effects of Data Sharing\label{wel}}

The model with quadratic payoffs (regardless of the prior distribution of
consumers' willingness to pay) yields explicit expressions for the value of
information for product market participants. In particular, because prices
and quantities are linear functions of the posterior mean $\widehat{w}_{i}$,
the ex ante average prices and quantities $\mathbb{E}\left[ p^{\ast }\right] 
$ and $\mathbb{E}\left[ q^{\ast }\right] $ are constant across all
information structures. Consequently, all surplus levels depend only on the
ex ante variance of the posterior mean $\widehat{w}_{i}$ under the
consumers' and the merchant's information structures.

We therefore quantity the players' information gains under any information
structure $Y$ through the variance of posterior expectation:%
\begin{equation}
G(Y)\triangleq \limfunc{var}[\widehat{w}_{i}(y)],  \label{gf}
\end{equation}%
and refer to it as the (informational) \emph{gain function}. As we
normalized the variance of the fundamental $w_{i}$ to $1$, the gain function 
$G(Y)$ represents the fraction of the variance of $w_{i}$ explained by the
signal $y$.

We now turn to the consequences of data sharing relative to no information
sharing. Without information, the producer charges a constant price for all
consumers based on the prior mean, denoted by $\bar{p}$. In contrast, the
consumer already has an initial signal $s_{i}$, according to which she can
adjust her quantity. The producer's net revenue and the consumer's expected
utility under no information sharing are then given by%
\begin{eqnarray*}
\Pi _{i}\left( S_{i},\varnothing \right) &\triangleq &\mathbb{E}\left[ \pi
\left( \bar{p},q_{i}^{\ast }\left( s_{i}\right) \right) \right] , \\
U_{i}\left( S_{i},\varnothing \right) &\triangleq &\mathbb{E}\left[
u_{i}\left( w_{i},q_{i}^{\ast }\left( s_{i}\right) ,\bar{p}\right)
\left\vert S_{i}\right. \right] .
\end{eqnarray*}

We can now express the value of data sharing for the consumers and the
producer in terms of the respective information gains.

\begin{proposition}[Value of Data Outflow]
\label{impact}\strut

\begin{enumerate}
\item The value of data outflow $Y$ for the producer is 
\begin{equation}
\Pi _{i}\left( \left( S_{i},Y_{i}\right) ,Y_{0}\right) -\Pi _{i}\left( S_{i},%
\mathcal{\varnothing }\right) =\frac{1}{4}G\left( Y_{0}\right) .
\label{deltaPi}
\end{equation}

\item The value of data outflow $Y$ for consumer $i$\ is 
\begin{equation}
U_{i}\left( \left( S_{i},Y_{i}\right) ,Y_{0}\right) -U_{i}\left(
S_{i},\varnothing \right) =\frac{1}{2}\left( G\left( S_{i},Y_{i}\right)
-G\left( S_{i}\right) \right) -\frac{3}{8}G\left( Y_{0}\right) .
\label{deltaU}
\end{equation}

\item The social value of data outflow $Y$ is 
\begin{equation}
W_{i}\left( \left( S_{i},Y_{i}\right) ,Y_{0}\right) -W_{i}\left(
S_{i},\varnothing \right) =\frac{1}{2}\left( G\left( S_{i},Y_{i}\right)
-G\left( S_{i}\right) \right) -\frac{1}{8}G\left( Y_{0}\right) .
\label{deltaW}
\end{equation}
\end{enumerate}
\end{proposition}

Thus, consumer and social surplus increase with the additional information
learned by the consumers $G\left( S_{i},Y_{i}\right) -G\left( S_{i}\right) $%
, and decrease with the information learned by the producer $G\left(
Y_{0}\right) $. Intuitively, the welfare consequences of data sharing
operate through two channels. First, with more information about her own
preferences, the demand of each consumer is more responsive to her
willingness to pay; this responsiveness is beneficial for consumers and
(weakly) for the producer. Second, with access to better data, the producer
adapts his pricing policy to the estimate of each consumer's willingness to
pay $\widehat{w}_{i}$. This price responsiveness dampens some of the
quantity responsiveness. Hence, this second channel reduces both consumer
surplus and total welfare.

Whether the first or the second channel dominates depends on the
informativeness of the consumers' initial signals $G\left( S_{i}\right) $,
and on the degree of correlation in the fundamental and error terms, which
jointly determine any consumer's ability to learn from others' information.
Corollary \ref{inef} formalizes this intuition by deriving the implications
of Proposition \ref{impact} in several special cases.

\begin{corollary}[Welfare Effects]
\label{inef}\strut

\begin{enumerate}
\item If consumers cannot learn from each other's signals $\left( G\left(
S\right) =G\left( S_{i}\right) \right) $, any data sharing reduces consumer
and social surplus.

\item If individual signals $s_{i}$ are uninformative $\left( G\left(
S_{i}\right) =0\right),$ any data sharing improves consumer and social
surplus.

\item Social surplus is maximized by collecting and sharing all signals with
every consumer $\left( Y_{i}=S\right) $, and sharing no data with the
producer $\left( Y_{0}=\varnothing \right) $.
\end{enumerate}
\end{corollary}

Data sharing is detrimental to consumer and social surplus when consumers
observe their willingness to pay perfectly $\left( \sigma =0\right) $, or
when both fundamentals $\left( w_{i},w_{j}\right) $ and errors $\left(
e_{i},e_{j}\right) $ are independent. In those cases, any data sharing only
enables price discrimination.\footnote{%
The special case in which each consumer knows her willingness to pay (i.e.,
signals are noiseless in our model's language) is closely related to the
model of third-degree price discrimination in \cite{robi33} and \cite{schm81}%
. In our setting, data sharing enables the producer to offer personalized
prices; thus, price discrimination occurs across different \emph{realizations%
} of the willingness to pay. In contrast, in \cite{robi33} and \cite{schm81}%
, price discrimination occurs across different market segments. In both
settings, the central result is that average demand does not change (with
all markets served), but social welfare is lower under finer market
segmentation.} Conversely, if individual signals become arbitrarily
uninformative, but the entire vector $s$ remains informative, then even 
\emph{symmetric }information gains (i.e., data outflow policies where $%
Y_{0}=Y_{i}$) yield Pareto improvements in the product market. In this case,
the producer and the consumer share the additional gains from trade
associated with better informed consumption and pricing decisions.

Finally, an immediate implication of the two channels highlighted by
Proposition \ref{impact} is that the \emph{first best }allocation of
information consists of collecting and sharing all data among the consumers
and none with the producer. However, a data intermediary with market power
will not implement the socially optimal allocation of information. In
particular, because the producer is paying for the data, he will always
receive some information. To fully describe the outcome of the game, we then
turn to the price of social data.

\subsection{Price of Social Data\label{psd}}

We first derive the total payment $m_{0}$ charged to the producer and the
compensation $m_{i}$ owed to each consumer. For the producer, the gains from
data acquisition have to at least offset the price of the data. At the same
time, the intermediary can charge up to the entire value of the information
outflow $Y_{0}$ to the producer. From expression (\ref{deltaPi}) in
Proposition \ref{impact}, we can then write the payment $m_{0}$ as%
\begin{equation}
m_{0}\left( Y\right) =N\left( \Pi _{i}\left( \left( S_{i},Y_{i}\right)
,Y_{0}\right) -\Pi _{i}\left( S_{i},\varnothing \right) \right) =\frac{N}{4}%
G\left( Y_{0}\right) .  \label{irf1}
\end{equation}%
Not surprisingly, the intermediary's profits are increasing in the amount of
information sold to the producer. However, we also know from Lemma \ref{tan}
that every consumer $i$ receives at least as much information as the
producer. These two observations establish the optimality of the \emph{%
complete data outflow }policy. Under this policy, the entire realized data
inflow $X$ is reported to the producer and to all consumers, including those
who did not accept the intermediary's offer.

\begin{lemma}[Optimal Data Outflow]
\label{dop}\quad \newline
Given any realized data inflow $X$, the complete data outflow policy, $%
Y_{0}^{\ast }\left( X\right) =Y_{i}^{\ast }\left( X\right) =X$\ for all $i$,$%
\ $maximizes the gross revenue of the producer among all feasible
data-outflow policies.
\end{lemma}

A critical driver of the consumer's decision to share data is her ability to
anticipate the intermediary's use of the information thus gained. By Lemma %
\ref{dop}, every consumer knows that all product market participants will
receive the same information from the intermediary. In particular, consumer $%
i$ knows that, by rejecting her contract, she prevents the producer from
accessing her data $X_{i}$, but she does not forego the opportunity to learn
from other consumers' data $X_{-i}$. In other words, consumer $i$ can learn $%
X_{-i}$ for free. In practice, this corresponds to {searching for the same
content on a digital platform without \textquotedblleft logging
in\textquotedblright\ or else agreeing to sharing her own data}.

For any data inflow policy $X$ and any underlying signal structure, the data
intermediary offers positive payments to consumers. This occurs because the
intermediary must compensate consumers \emph{on the margin}:\emph{\ }%
consumer $i$ requires compensation for revealing her data $X_{i}$ to the
producer, given that the other $N-1$ consumers already share theirs.%
\footnote{%
We will revisit this property when we discuss the intermediary's commitment
power in Section \ref{iwc}.} Therefore, the payments satisfy%
\begin{equation}
m_{i}\geq U_{i}\left( \left( S_{i},X_{-i}\right) ,X_{-i}\right) -U_{i}\left(
\left( S_{i},X\right) ,X\right) .  \label{ir}
\end{equation}%
Intuitively, consumer $i$ is not compensated for the (positive or negative)
effect of other consumers' data inflow $X_{-i}$ on her surplus. To see this
more formally, suppose consumer $i$'s participation constraint (\ref{ir})
binds, and rewrite compensation $m_{i}$ as%
\begin{eqnarray}
m_{i}^{\ast }\left( X\right) &=&U_{i}\left( \left( S_{i},X_{-i}\right)
,X_{-i}\right) -U_{i}\left( \left( S_{i},X\right) ,X\right)  \notag \\
&=&-\underset{\triangleq \Delta U_{i}\left( X\right) }{\underbrace{\left(
U_{i}\left( \left( S_{i},X\right) ,X\right) -U_{i}\left( S_{i},\varnothing
\right) \right) }}+\underset{\triangleq DE_{i}\left( X\right) }{\underbrace{%
U_{i}\left( \left( S_{i},X_{-i}\right) ,X_{-i}\right) -U_{i}\left(
S_{i},\varnothing \right) }}.  \label{mi}
\end{eqnarray}%
The first term in (\ref{mi}), denoted by $\Delta U_{i}\left( X\right) $, is
the total impact on consumer $i$'s surplus associated with data inflow $X$.
The second term is the \emph{data externality }(\ref{de1})$\ $imposed on $i$
by consumers $j\neq i$. It reflects the change in utility when $j\neq i$
sell their data $X_{j}$ to the intermediary who then shares the data with
the producer. As it is central to our analysis, we now examine the latter
term in detail.

\subsection{Data Externality and Intermediation\label{de_sec}}

Our notion of \emph{data externality} isolates the effect on consumer $i$'s
surplus of the decision by the other consumers to share their data with all
market participants.

\begin{definition}[Data Externality]
\label{dedef}\quad \newline
The data externality imposed by consumers $-i$ on consumer $i$ is given by 
\begin{equation}
DE_{i}\left( X\right) \triangleq U_{i}\left( \left( S_{i},X_{-i}\right)
,X_{-i}\right) -U_{i}\left( S_{i},\varnothing \right) .  \label{de1}
\end{equation}
\end{definition}

Using expression (\ref{deltaU}) in Proposition \ref{impact}, the data
externality can be written as%
\begin{equation}
DE_{i}\left( X\right) =\frac{1}{2}\left( G\left( S_{i},X_{-i}\right)
-G\left( S_{i}\right) \right) -\frac{3}{8}G\left( X_{-i}\right) .
\label{de2}
\end{equation}%
To provide some intuition as to what determines the sign of the data
externality, we evaluate expression (\ref{de2}) under the two special
information structures in Examples \ref{e1} and \ref{e2}. In both cases, we
consider what would happen if consumer $i$ held back her signal, given that
the remaining $N-1$ consumers share theirs with the producer and with
consumer $i$\emph{.}

Example \ref{e1} illustrates that if consumer $i$ does not learn much from
the signals of the other consumers, but those signals help predict $w_{i}$,
then the data externality is negative.

\begin{example2}[Common Preferences]
\quad \newline
Let $s_{i}=w+\sigma e_{i}$, and suppose the intermediary collects all
consumer data, $X_{i}=S_{i}$. The producer can use $N-1$ signals to estimate
the common preference parameter $w$, i.e., $G\left( S_{-i}\right) >0$. If
the individual signals are sufficiently precise so that $G\left( S\right)
\approx G\left( S_{i}\right) $, then consumer $i$ is worse off when other
consumers share their signals, i.e., the data externality is negative.
\end{example2}

Example \ref{e2} illustrates that if the producer cannot learn anything
about $w_{i}$ from signals $s_{-i}$, then the data externality is
unambiguously positive.

\begin{example2}[Common Experience]
\quad \newline
Let $s_{i}=w_{i}+\sigma e$, and suppose the intermediary collects all
consumer data, $X_{i}=S_{i}$. Because all $w_{i}$ are independent, $G\left(
S_{-i}\right) =0$, i.e., the producer cannot learn about $w_{i}$ from
signals $s_{-i}$ only. However, consumer $i$ can use signals $s_{-i}$ to
filter out the common error in her own signal $s_{i}$, i.e., $G\left(
S\right) >G\left( S_{i}\right) $. Therefore, other consumers' signals help
consumer $i$, i.e., the data externality is positive.
\end{example2}

Thus, the overall effect of data sharing on consumer surplus (Proposition %
\ref{impact}) depends largely on the informativeness of individual signals $%
s_{i}$. Conversely, the impact of other consumers' sharing decisions varies
significantly with the data structure, particularly the ability of the
producer to learn about $w_{i}$ from signals $s_{-i}$.

\section{Optimal Data Intermediation\label{opdi}}

The data externality has direct implications for the intermediary's profit (%
\ref{obj}). Combining the expressions for payments in (\ref{irf1}) and (\ref%
{mi}), we can write the intermediary's profit $R$ as%
\begin{equation}
R\left( X\right) =\sum\nolimits_{i=1}^{N}\left( \Delta W_{i}\left( X\right)
-DE_{i}\left( X\right) \right) \text{,}  \label{RS}
\end{equation}%
where $\Delta W\left( X\right) $ denotes the effect of sharing data policy $%
X $ on total surplus, as in (\ref{deltaW}). The intermediary's profits are
equal to the effect of data sharing on social surplus, net of the data
externalities across all consumers. The sign of the data externality is
therefore critical for the profitability of data intermediation. In
particular, if consumers impose negative data externalities on each other,
this imposition directly reduces the compensation owed to each one, and
conversely if the data externalities are positive. The revenue formula (\ref%
{RS}) clarifies how the intermediary's objective differs from the social
planner's. In particular, if the data externality is negative, then
intermediation can be profitable but welfare reducing. Conversely, if the
data externality is positive, welfare-enhancing intermediation might not be
profitable.

Having characterized the two terms $\Delta U_{i}$ and $DE_{i}$ in (\ref%
{deltaU}) and (\ref{de2}), we can rewrite the payments to consumers in (\ref%
{mi}) as%
\begin{equation}
m_{i}^{\ast }\left( X\right) =\frac{3}{8}\left( G\left( X\right) -G\left(
X_{-i}\right) \right) .  \label{mistar}
\end{equation}%
Finally, combining the terms $m_{0}^{\ast }\left( X\right) $ in (\ref{irf1})
and $m_{i}^{\ast }\left( X\right) $ in (\ref{mistar}), we obtain%
\begin{equation}
R\left( X\right) =\frac{N}{8}\left( 3G\left( X_{-i}\right) -G\left( X\right)
\right) .  \label{revx}
\end{equation}%
This yields a necessary and sufficient condition for profitable data
intermediation.

\begin{proposition}[Profitable Data Intermediation]
\label{prof}\qquad \newline
Data intermediation with inflow policy $X$ is profitable if and only if 
\begin{equation*}
3G\left( X_{-i}\right) \geq G(X).
\end{equation*}
\end{proposition}

Proposition \ref{prof} considers the amount of information learned by the
producer. Specifically, it requires that the signals $x_{-i}$ generate at
least $1/3$ of the variance of $w_{i}$ explained by the entire vector $x$ in
order for the data inflow policy $X$ to be profitable. Intuitively, it is
cheaper to acquire each signal $x_{i}$ if the other consumers' signals are
substitutes, and this is more likely to occur when the underlying
fundamentals $w_{i}$ and $w_{-i}$ are correlated. Conversely, for
independent fundamentals, $G\left( X_{-i}\right) =0$ for any $X$, and
intermediation is not profitable.

We now draw the implications for the intermediary's profits in our two
leading examples.

\begin{corollary}[Common Preference]
\label{IbP}\qquad \newline
Suppose fundamentals $w_{i}$ are perfectly correlated and errors $e_{i}$ are
independent. When $N$ is large, data intermediation is always profitable.
However, for sufficiently small $\sigma $, the data externality is negative,
and data sharing reduces social surplus.
\end{corollary}

These results echo the findings of \cite{ammo21}, who considered signals
with diminishing marginal informativeness, and found socially excessive data
intermediation. The information structure in Corollary \ref{IbP} satisfies
this submodularity property. In our model, however, socially insufficient
intermediation can also occur. In particular, the intermediary may be unable
to generate positive profits from socially efficient information with
complementary signals, such as those in Corollary \ref{EbU}.

\begin{corollary}[Common Experience]
\label{EbU}\qquad \newline
Suppose fundamentals $w_{i}$ are independent and errors $e_{i}$ are
perfectly correlated. For sufficiently large $\sigma $, data sharing
increases social surplus. However, because the fundamentals are independent,
data intermediation is never profitable.
\end{corollary}

The conclusions of Corollaries \ref{IbP} and \ref{EbU} do not depend on
whether the intermediary collects complete data or anonymized data. However,
as we discuss in the next section, it is always optimal for the intermediary
to anonymize the data collected.

\subsection{Data Anonymization\label{daan}}

We now explore the intermediary's decision to anonymize the individual
consumers' demand data. We focus on two maximally different policies along
this dimension. At one extreme, the intermediary can collect and transmit 
\emph{complete (identity-revealing) }data about individual consumers $\left(
X=S\right) $, thereby enabling the producer to charge personalized prices.
At the other extreme, the intermediary can collect \emph{anonymized} data $%
\left( X=A\right) $.

Under anonymized information intermediation, the producer charges the same
price to all consumers who participate in the intermediary's data policy. In
other words, from the point of view of the producer, anonymized data is
equivalent to aggregate demand data. These data still allow the producer to
perform third-degree price discrimination across realizations of the total
market demand but limit his ability to extract surplus from individual
consumers.\footnote{%
More formally, under the anonymized data policy $A$, the producer has access
to the vector $\delta \left( s\right) $, i.e., to a uniformly random
permutation of the consumers' signals. Because the producer faces a
prediction problem for each $w_{i}$ with a convex loss, he chooses to charge
a uniform price that is optimal for the sample average of the consumers'
signals.}

Certainly, for the producer, the value of market demand data is lower than
the value of individual demand data. However, the cost of acquiring such
fine-grained data from consumers is also correspondingly higher. We now show
that anonymizing the consumers' information profitably reduces the
intermediary's data acquisition costs.

\begin{proposition}[Optimality of Data Anonymization]
\label{si}\quad \newline
The intermediary obtains strictly greater profits by collecting anonymized
consumer data.
\end{proposition}

Within the confines of our policies, but independent of the distributions of
fundamental and noise terms, the data intermediary finds it advantageous to
not elicit the identity of the consumer. Therefore, the producer will not
offer personalized prices but variable prices that adjust to the realized
information about market demand. In other words, a monopolist intermediary
might cause socially inefficient information transmission, but the
equilibrium contractual outcome preserves privacy over the personal identity
of the consumer.

In Section \ref{limits}, we explore the boundaries of the anonymization
result, under both heterogeneous consumers and alternative product-market
specifications. In particular, we will generalize our result to show that
anonymization is optimal when consumers are homogeneous ex ante and
transmitting data to the producer is socially inefficient. Therefore, if
information is used to target advertisements, and more precise messages
increase social surplus, then we shall predict that the intermediary shares
complete data.\footnote{%
If, in addition, responsiveness to ads is heterogeneous in the consumer
population, coarser information transmission (i.e., not fully anonymous) can
also be optimal, as we show in Proposition \ref{gs}.}

For the case of surplus extraction (prices), the finding in Proposition \ref%
{si} suggests why we might see personalized prices in fewer settings than
initially anticipated. In the context of direct sales of information, for
example, Nielsen does not sell individual households' data to merchants.
Instead, Nielsen aggregates its panel data at the local market level.
Similarly, in the context of indirect sales of information, merchants on the
retail platform Amazon very rarely engage in personalized pricing. However,
the price of every single good or service is subject to substantial
variation across both geographic markets and over time.

In light of the above result, we might interpret the restraint on the use of
personalized pricing as the optimal resolution of the intermediary's
trade-off in acquiring sensitive consumer information. Stretching the
interpretation, the pushback against Amazon's introduction of personalized
pricing in the early 2000s can also be viewed as a tightening the consumers'
participation constraints. Under any degree of competition, the net value of
the Amazon platform under perfect personalized pricing would not have
exceeded the consumers' outside options.

The data externality is, once again, the key to gaining intuition for why
the intermediary chooses data anonymization. Suppose consumers $-i$ reveal
their signals, and consumer $i$ does not. With access to identifying
information, the producer optimally aggregates the available data to form
the best predictor of the missing data point. In this case, the producer
charges a personalized price $p_{i}^{\ast }\left( X_{-i}\right) $ to each
nonparticipating consumer $i$. With anonymous data, the producer charges two
prices: a single price for all participating consumers and another price for
the deviating, nonparticipating consumers. Because the distribution of
consumer willingness to pay and signals is symmetric, however, the
producer's inference on $w_{i}$ is invariant to permutations of the other
consumers' signals, i.e.,%
\begin{equation*}
\mathbb{E}\left[ w_{i}\mid a_{-i}\right] =\mathbb{E}\left[ w_{i}\mid s_{-i}%
\right] .
\end{equation*}%
Therefore, a nonparticipating consumer faces identical prices under both
data policies:\footnote{%
Anonymization remains optimal if we force the producer to charge a single
price to all consumers on and off the equilibrium. With this interpretation,
we intend to capture the idea that the producer offers one price
\textquotedblleft on the platform\textquotedblright\ to the participating
consumers while interacting with the deviating consumer \textquotedblleft
offline.\textquotedblright\ The producer then uses the available market data
to tailor the offline price.}%
\begin{equation*}
p_{i}^{\ast }\left( S_{-i}\right) =p_{i}^{\ast }\left( A_{-i}\right) .
\end{equation*}%
Likewise, consumer $i$'s posterior distribution over her own $w_{i}$ does
not depend on the identity the other consumers. Therefore, removing identity
information through the anonymized policy $X=A$ does not have any
implications for consumers' learning either:%
\begin{equation*}
\mathbb{E}\left[ w_{i}\mid s_{i},a_{-i}\right] =\mathbb{E}\left[ w_{i}\mid s%
\right] .
\end{equation*}

Because the amount of information available to consumer $i$ and to the
producer \emph{off the path of play }is independent of $X\in \left\{
A,S\right\} $, it follows that%
\begin{equation*}
U_{i}\left( \left( S_{i},S_{-i}\right) ,S_{-i}\right) =U_{i}\left( \left(
S_{i},A_{-i}\right) ,A_{-i}\right) .
\end{equation*}%
In turn, this implies $DE_{i}\left( S\right) =DE_{i}\left( A\right) $. Thus,
the data externality term $DE_{i}\left( X\right) $ in the intermediary's
profits (\ref{RS}) is not impacted by the choice of inflow $X\in \left\{
A,S\right\} .$

Along the path of play, however, the two data inflow policies yield
different outcomes. In particular, the anonymized data inflow policy reduces
the amount of information conveyed to the producer in equilibrium.
Crucially, this reduction does not occur at the expense of the consumers'
own learning. Therefore, the shift to anonymized data increases the total
surplus terms $\Delta W_{i}\left( X\right) $ and the intermediary's profits.
Put differently, anonymization reduces the cost of procuring the
information, relative to the loss in revenue.

We now show that data anonymization is the key to the \textquotedblleft
explosive\textquotedblright\ profitability of data intermediation when the
number of consumers becomes large. We then revisit the applications and
limitations of our anonymization result in Section \ref{limits}.

\subsection{Large Markets\label{lama}}

Thus far, we have considered the optimal data policy for a given finite
number of consumers, each of whom transmits a single signal. Perhaps,\emph{\
the} defining feature of data markets is the multitude of (potential)
participants, data sources, and services. We now pursue the implications of
having many participants (i.e., of many data sources) for the social
efficiency of data markets and the price of data.

Each additional consumer presents an additional opportunity for trade in the
product market. Thus, the feasible social surplus is linear in the number of
consumers. In addition, with every additional consumer, the intermediary
obtains additional information about market demand. These two effects
suggest that intermediation becomes increasingly profitable in larger
markets, wherein the potential revenue increases without bound, whereas
individual consumers make a small marginal contribution to the precision of
aggregate data.

For this comparative statics analysis, we adopt the following \emph{additive
data structure. }Specifically, we assume the willingness to pay of consumer $%
i\ $is the sum of two components:%
\begin{equation}
w_{i}=\theta +\theta _{i}.  \label{wa}
\end{equation}%
The term $\theta $ is \emph{common }to all consumers in the market, whereas
the term $\theta _{i}$ is \emph{idiosyncratic }to consumer $i$. Similarly,
the error term of consumer $i$ is given by%
\begin{equation}
e_{i}\triangleq \varepsilon +\varepsilon _{i},  \label{ea}
\end{equation}%
where the terms $\varepsilon $ and $\varepsilon _{i}$ refer to a common and
an idiosyncratic error, respectively. We also refer to the willingness to
pay $w_{i}$ as the fundamental as opposed to the error term $e_{i}$.

As we vary the number of consumers $N$, the additive data structure allows
us to hold the pairwise correlation between any two consumers' fundamentals
and noise terms constant. In particular, let $\alpha $ denote the
correlation coefficient of any two $\left( w_{i},w_{j}\right) $, and let $%
\beta $ denote the correlation coefficient of $\left( e_{i},e_{j}\right) $.

We first establish a sufficient condition for the profitability of complete
data sharing as the number of consumers becomes large, and then we analyze
the data intermediary's revenue and total cost separately.\newpage 

\begin{proposition}[Profitable Intermediation of Anonymized Data]
\label{prof_anonym} \qquad \newline
For any $\alpha >0$, there exists $N^{\ast }$ such that anonymized data
sharing is profitable if $N>N^{\ast }.$
\end{proposition}

We already know from Corollary \ref{IbP} that a high degree of correlation
in the consumers' willingness to pay allows the intermediary to profit from
data sharing with sufficiently precise signals. Under the optimal
data-sharing policy, \emph{any }degree of correlation in the consumers'
willingness to pay makes the anonymized signals sufficiently close
substitutes that intermediation is profitable when $N$ is large.

In Proposition \ref{lm}, we assume that error terms are independent. This
allows us to use the sample average to establish a lower bound on learning
from $N-1$ signals. We suspect that similar results hold more generally
under correlated errors.\footnote{%
This result holds, for example, when both fundamentals and errors have
Gaussian distributions.}

\begin{proposition}[Large Markets]
\label{lm}\qquad \newline
Consider the additive data structure and assume that errors are independent
across consumers. As $N\rightarrow \infty $:

\begin{enumerate}
\item \label{lm1}Each consumer's compensation $m_{i}^{\ast }$ converges to
zero.

\item \label{lm2}Total consumer compensation is bounded by a constant, 
\begin{equation*}
Nm_{i}^{\ast }\leq \frac{9}{8}\left( \limfunc{var}\left[ \theta _{i}\right] +%
\limfunc{var}\left[ \varepsilon _{i}\right] \right) ,\quad \forall N.
\end{equation*}

\item \label{lm3}The intermediary's revenue and profit grow linearly in $N$.
\end{enumerate}
\end{proposition}

As the optimal data policy aggregates the consumers' signals, each
additional consumer has a rapidly decreasing marginal value. Furthermore,
each consumer is paid only for her marginal contribution; this explains why
the total payments $Nm_{i}$ converge to a finite number. Strikingly, this
convergence can occur from above$\emph{:}$ when the consumers' willingness
to pay is sufficiently correlated, the decrease in each $i$'s marginal
contribution can be sufficiently strong to offset the increase in $N$.

Whereas total costs converge to a constant, the revenue that the data
intermediary can extract from the producer is linear in the number of
consumers. Our model therefore implies that, as the market size grows
without bound, the per capita profit of the data intermediary converges to
the per capita profit when the (anonymized) data are freely available.
Conversely, the impact on consumer surplus depends on the degree of
correlation in the underlying fundamentals and on the precision of the
consumers' initial signals.\footnote{%
In a recent contribution, \cite{loma20} study large digital monopoly
markets, where data have the countervailing effects of improving consumer
valuations and increasing monopoly prices.}

Finally, we show that data anonymization is crucial for the large $N$
properties of the intermediary's profits. Recall that, with complete data
intermediation, individual consumer payments are proportional to $%
G(S)-G(S_{-i})$. As long as fundamentals $w_{i}$ are not perfectly
correlated, these payments are then bounded away from zero for any finite $N$%
. Proposition \ref{limrev} shows that this property also holds in the limit.

\begin{proposition}[Asymptotics with Complete Sharing]
\label{limrev}\qquad \newline
Consider the additive data structure with $\limfunc{var}[\theta _{i}]>0$.
Under complete (identity-revealing) data sharing, the asymptotic individual
compensation is bounded away from $0$: 
\begin{equation*}
\liminf_{N\rightarrow \infty }m_{i}^{\ast }\geq \frac{3}{8}\frac{\limfunc{var%
}^{2}[\theta _{i}]}{1+\limfunc{var}[e_{i}]}>0.
\end{equation*}
\end{proposition}

An immediate consequence of Proposition \ref{limrev} is that, with complete
data sharing, total payments to consumers grow linearly in $N$. Thus,
anonymization is critical to achieving increasing returns to scale in data
intermediation: even in settings where complete data intermediation $X=S$ is
profitable, the per capita profits are bounded away from the full value of
information.

Figure \ref{compensation} illustrates an example with normally distributed
fundamentals and errors, in which it can be less expensive for the
intermediary to acquire a larger \emph{anonymized }dataset than a smaller
one, but not a larger \emph{complete }dataset.

\begin{figure}
	\label{compensation}%
	\centering
	\caption{Total Consumer Compensation\ $(\sigma_w=1, \sigma_e=0)$}%
	\includegraphics[width=0.6\textwidth]{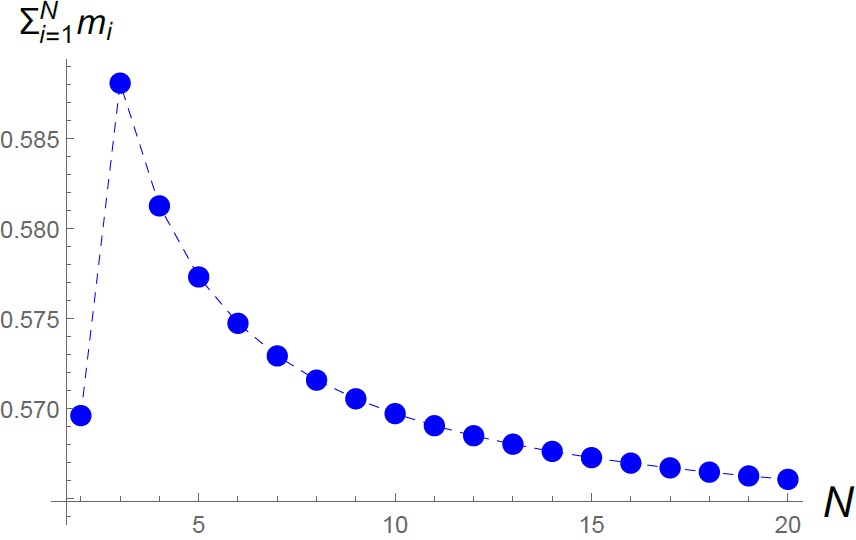}
\end{figure}

%TCIMACRO{\TeXButton{B}{\begin{figure}[htbp]\centering}}%

\subsection{Unique Implementation\label{dac}}

Our analysis thus far has characterized the intermediary's most preferred
equilibrium. An ensuing question is whether the qualitative insights and the
asymptotic properties discussed above would hold across all equilibria,
particularly in the intermediary's least preferred equilibrium. A seminal
result in the literature on contracting with externalities (see \cite{sega99}%
) is the \textquotedblleft divide-and-conquer\textquotedblright\ scheme that
guarantees a unique equilibrium outcome (see \cite{sewh00} and \cite{mish16}%
). Under this scheme, the intermediary can sequentially approach consumers
and offer compensation conditional on all earlier consumers having accepted
an offer. In this scheme, the first consumer receives compensation equal to
her entire surplus loss, thereby guaranteeing her acceptance regardless of
the other consumers' decisions. More generally, consumer $i$ receives the
optimal compensation level in the baseline equilibrium when $N=i$.

The cost of acquiring the consumers' data is strictly higher under
\textquotedblleft divide and conquer\textquotedblright\ than in the
intermediary's most preferred equilibrium. Nonetheless, the impact of the
ensuring unique implementation on per capita profits vanishes in the limit.

\begin{proposition}[\textquotedblleft Divide and Conquer\textquotedblright ]

\label{doc}\qquad \newline
Consider the additive data structure with independent errors. Under the
\textquotedblleft Divide and Conquer\textquotedblright\ scheme, total
consumer compensation satisfies 
\begin{equation*}
Nm_{i}^{\ast }\leq \frac{3}{4}(1+\log N)(\limfunc{var}[\theta _{i}]+\limfunc{%
var}[\varepsilon _{i}]).
\end{equation*}
\end{proposition}

Under divide and conquer, the total payments to the consumers do not
converge to a finite constant as $N$ grows without bound. However, the
growth rate of these payments is far smaller than the rate at which the
producer's willingness to pay for data diverges. Therefore, regardless of
the equilibrium-selection criterion, the intermediary's per capita profits
converge to the benchmark level when anonymized consumer data are made
available.

\section{Implications for Consumer Privacy\label{limits}}

In this section, we enrich our model along several dimensions to
characterize the implications of the optimal data intermediation policy for
consumer privacy. In particular, we consider richer pricing instruments in
the product market, heterogeneous consumers, heterogeneous product
varieties, noisy information collection, and commitment power in the data
market.

\subsection{Data Anonymization and Social Efficiency\label{gasec}}

In our baseline setting, data anonymization is optimal independent of the
model parameters, such as the number of consumers or the distribution of
fundamentals and error terms. This result relies on two crucial
assumptions:\ (i) consumer are homogeneous, by which we mean that the
distributions of fundamentals $w$ and errors $e$ are symmetric, and (ii)
data sharing has unambiguous welfare effects on product market participants.
Indeed, in the model of linear price discrimination, transmitting $X_{i}$
anonymously improves consumer and social surplus, relative to complete data
intermediation.

We can generalize this insight to \emph{any arbitrary} product market
interaction beyond the linear pricing model of the previous section.
Proposition \ref{ga} generalizes the intuition behind the optimal
anonymization result in Proposition \ref{si}: it establishes that \emph{%
social surplus }is the criterion guiding the intermediary's decision to
optimally collect anonymized data.

\begin{proposition}[Social Optimality of Data Anonymization]
\label{ga}\qquad \newline
With homogeneous consumers, anonymized data intermediation is more
profitable than complete data intermediation if and only if anonymization
increases social surplus.
\end{proposition}

We note two important aspects of this result. First, it establishes the
congruence between the intermediary private objective and the social welfare
with respect to the anonymization decision \emph{only}. It does not claim
that the equilibrium information flow itself is socially efficient. Second,
the argument does not require any specific feature of the product market
interaction. As the decision between anonymization and de-anonymization
pertains precisely to the marginal value of the private information of $i$
for the prediction of the willingness to pay $w_{i}$, intermediary and
consumer $i$ can attain a socially efficient arrangement.

The result has immediate implications for how equilibrium data sharing
policies depend on the nature of the product market interaction. In
particular, Proposition \ref{ga} allows us to examine the role of richer
pricing instruments. \cite{bebm15} have shown that every feasible
combination of consumer and producer surplus is consistent with \emph{some }%
form of price discrimination. Proposition \ref{ga} shows that, if the
producer had the ability to extract all the expected surplus (given the
consumers' information), then the intermediary would find it more profitable
to collect complete, identifying data.

A canonical example where this prediction is relevant is the case of unit
demand by the consumers, where our model would predict the prevalence of
perfect price discrimination. However, to the extent to which consumers have
options to retain some surplus ex post, such as by scaling down their
purchase level as in our baseline model, then full surplus extraction would
require the producer to have access to more complex pricing mechanisms.

\subsection{Market Segmentation and Data\label{segm}}

The assumption of ex ante homogeneity among consumers has enabled us to
produce some of the central implications of social data. A more complete
description of consumer demand should introduce heterogeneity across groups
of consumers along characteristics such as location, demographics, income,
and wealth.

We now explore how these additional characteristics influence information
policy and the profits of the data intermediary. To this end, we augment the
description of consumer demand by splitting the population into $J$
homogeneous groups:%
\begin{equation*}
w_{ij}\sim F_{w,j},~e_{i,j}\sim F_{e,j},~i=1,2,...,N_{j},~j=1,...,J\text{.}
\end{equation*}

The intermediary's data inflow policy must now specify whether to anonymize
the consumers' signals across groups and within each group. However,
Proposition \ref{ga} establishes that it is always more profitable to
anonymize all signals within each group, rather than revealing the
consumers' identities.

\begin{corollary}[No Discrimination within Groups]
\label{group_discri} \quad \newline
The data policy that anonymizes all signals within each group $j=1,...,J$
and only reveals the group identity of each consumer $i$ is more profitable
than the complete data-sharing policy.
\end{corollary}

By further specifying the model, we can identify conditions under which the
data intermediary will collect and transmit group characteristics. By
collecting information about the group characteristics, the intermediary
influences the extent of price discrimination. For example, the intermediary
could anonymize all signals across groups, thus forcing the producer to
offer only a single price. Alternatively, the intermediary could allow the
producer to discriminate between two groups of consumers by recording and
transmitting the group identities. As intuition would suggest, enabling
price discrimination across groups not only allows the intermediary to
charge a higher fee to the producer but also increases the compensation owed
to consumers.

Proposition \ref{gs} below sheds light on the optimal resolution of this
trade-off. In this result, we restrict attention to the case of symmetric
groups ($N_{j}=N$ for all $j$), with the additive data structure $%
w_{i}=\theta +\theta _{i}$, and independent noise terms in the consumers'
signals.

\begin{proposition}[Segmentation]
\label{gs}~\newline
If $N$ is large enough, inducing group-level pricing is more profitable for
the intermediary than inducing uniform pricing.
\end{proposition}

Whereas Proposition \ref{si} stated that the intermediary will not reveal
any information about consumer identity, Proposition \ref{gs} refines that
result: if the market is sufficiently large, then the intermediary will
convey limited identity information, i.e., each consumer's group identity.
This policy allows the producer to price discriminate across, but not
within, groups. Conversely, if the producer faces few consumers and their
willingness to pay are not highly correlated, then pooling all signals
reduces the cost of sourcing the data.

The limited amount of price discrimination, which operates optimally at the
group level rather than the individual level, can explain the behavior of
many platforms. For example, Uber and Amazon claim that they do not
discriminate at the individual level, but they condition prices on location,
time, and other dimensions that capture group characteristics.

The result in Proposition \ref{gs} is perhaps the sharpest manifestation of
the value of big data. By enabling the producer to adopt a richer pricing
model, a larger database allows the intermediary to extract more surplus.
Our result also clarifies the appetite of the platforms for large datasets:
because having more consumers allows the platform to profitably segment the
market more precisely, the value of the marginal consumer $i=N$ to the
intermediary remains large even as $N$ grows. In other words, allowing the
producer to segment the market is akin to paying a fixed cost (i.e., higher
compensation to the current consumers) to access a better technology (i.e.,
one that scales more easily with $N$). Figure \ref{mar}\ illustrates this
result for an example with normally distributed fundamentals and errors.

\begin{figure}
	\label{mar}%
	\centering
	\caption{Marginal Value of an Additional Consumer}%
	\includegraphics[width=0.6\textwidth]{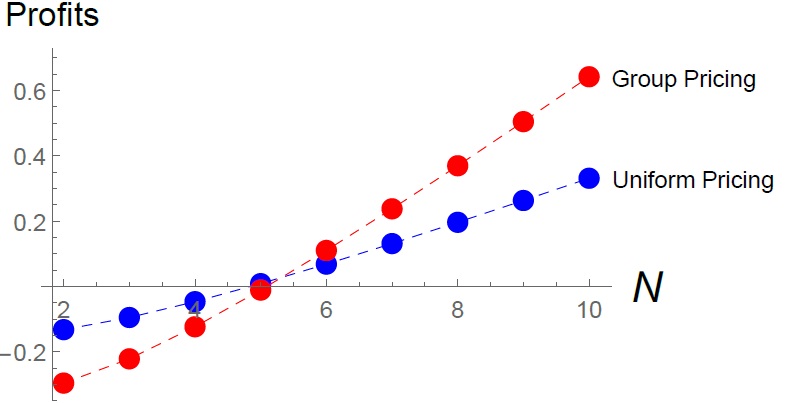}
\end{figure}

%EndExpansion

The optimality of using a richer pricing model when larger datasets are
available is reminiscent of model selection criteria under overfitting
concerns, e.g., the Akaike information criterion. In our setting, however,
the optimality of inducing segmentation is not driven by econometric
considerations. Instead, it is entirely driven by the intermediary's
cost-benefit analysis in acquiring more precise information from consumers.
As the data externality grows sufficiently strong, acquiring the data
becomes cheaper as the intermediary exploits the richer structure of
consumer demand.\footnote{\cite{oopp19} offer a demand-side explanation of a
similar phenomenon: they showed that data buyers who employ a richer pricing
model are willing to pay more for larger datasets.}

Finally, the welfare ranking of the two pricing schemes (group vs. uniform)
is a priori ambiguous. In particular, group pricing can yield lower total
surplus than uniform pricing, e.g., when consumers know their willingness to
pay and group pricing results in inefficient price discrimination. In other
words, outside of the conditions of Proposition \ref{ga} (e.g., with
heterogeneous consumers), market segmentation can be driven purely by the
data externality.

\subsection{Recommender System\label{rs}}

In our baseline model, the data shared by the intermediary are used by the
producer to set prices and by consumers to learn about their own
preferences. The first assumption is, in a sense, the worst-case scenario
for the intermediary: consider the case in which consumers' initial signals
are very precise. As price discrimination reduces total surplus, no
intermediation would be profitable without a strong, negative data
externality. Consequently, data aggregation is an essential part of the
optimal data intermediation policy in this case. In practice, however,
consumer data can also be used by the producer in surplus-enhancing ways,
for example, to facilitate targeting quality levels and other product
characteristics to the consumer's tastes.

In this section, we develop a generalization of our framework; this
generalization allows the producer to charge a unit price $p_{i}$ and to
offer a product of characteristic $k_{i}$ to each consumer. Consumers differ
both in their vertical willingness to pay and in their horizontal taste for
the product's characteristics. Consumer $i$'s utility function is given by%
\begin{equation*}
u_{i}\left( w_{i},q_{i},p_{i},k_{i},\ell _{i}\right) =(w_{i}-\left(
k_{i}-\ell _{i}\right) ^{2}-p_{i})q_{i}-q_{i}^{2}/2,
\end{equation*}%
with $w_{i}$ denoting the consumer's willingness to pay and $\ell _{i}$
denoting the consumer's ideal location or product characteristic. Both the
willingness to pay $w\in \mathbb{R}^{N}$ and the locations $\ell \in \mathbb{%
R}^{N}$ of different consumers are potentially correlated. The producer has
a constant marginal cost of quantity provision that we normalize to zero and
can freely set the product's characteristic. Therefore, the case of a common
location $\ell _{i}\equiv \ell $ for all consumers yields the baseline model
of price discrimination.

We examine the data intermediary's optimal data inflow policy, which allows
for separate aggregation policies for willingness to pay and location
information. We impose the following assumptions:\ (i) the gains from trade
under no information sharing are sufficiently large; (ii) the consumers'
fundamentals have a joint Gaussian distribution; and (iii) consumer $i$
perfectly observes $\left( w_{i},\ell _{i}\right) $. The extension to noisy
Gaussian signals is immediate. We then obtain another application of
Proposition \ref{ga}.

\begin{proposition}[Optimal Aggregation by a Recommender System]
\label{vag3}\qquad \newline
The intermediary's optimal policy collects anonymized data on the vertical
component $w_{i}$ and complete data on the horizontal component $\ell _{i}$.
\end{proposition}

Therefore, the recommender system enables the producer to offer targeted
product characteristics that match $k_{i}$ to $\ell _{i}$ as closely as
possible. However, the system does not allow for personalized pricing. The
logic is once again given by the intermediary's sources of profits, i.e.,
the contribution to social welfare $\Delta W$ and the data externality $DE$.
Because the data externalities do not depend on the level of data
anonymization, the intermediary chooses to aggregate the vertical dimension
of consumer data, thereby reducing the total surplus if transmitted to the
producer. Conversely, because the distance between a consumer's ideal
product and the firm's offer $\left( k_{i}-\ell _{i}\right) ^{2}$\ shifts
the consumer's demand function down, the intermediary allows for the
personalization of product characteristics.

\subsection{Intermediary with Commitment\label{iwc}}

We have assumed thus far that the intermediary cannot refrain from selling
information to the producer and cannot sell any acquired data inflow back to
the consumers. The latter assumption entails no loss: consumers know that
the intermediary sells the data to the producer, and therefore expect to
receive all available information regardless of their participation
decisions (Proposition \ref{tan}), because they know that revealing this
information maximizes the producer's fee.

The no commitment assumption reflects the substantial control that large
online platforms have over the use of the data and the opacity with which
the data outflow is linked to the data inflow. In other words, it is
difficult to ascertain how any given data input informs an intermediary's
data output. Nonetheless, it is useful to consider the implications of the
data intermediary's ability to commit to a certain data policy, especially
in light of the welfare properties of data sharing discussed above. To that
end, suppose the data intermediary offered the consumers contracts that
specify a data inflow \emph{and }a data outflow policy.

Through richer contracts, the data intermediary can offer consumers privacy
guarantees. In particular, the intermediary can implement the socially
efficient data-sharing policy, which consists of sharing all signals among
all the consumers who accept the contract and not sharing any data at all
with the producer (Corollary \ref{inef}). In exchange for this commitment,
the data intermediary requests compensation from the consumers. In turn,
consumers are willing to pay a positive price for these data, and hence the
socially optimal data sharing is always profitable.\footnote{%
This environment with commitment is related to the analysis in \cite{lizz99}
but has a number of distinct features. First, in \cite{lizz99}, the private
information is held by a single agent, and multiple downstream firms compete
for the information and for the object offered by the agent. Second, the
privately informed agent enters the contract after she has observed her
private information. The shared insight is that the intermediary \emph{with }
commitment power might be able to extract a rent without any influence on
the efficiency of the allocation.}

However, the equilibrium outcome under these stronger commitment assumptions
fails to capture the role of large online platforms. Even though there are
examples in which consumers pay a positive price to access tailored,
non-sponsored recommendations, data intermediaries choose to monetize the
producers' side of their platform much more frequently. Moreover, the
socially efficient policy need not maximize the intermediary's profits. For
example, if fundamentals are perfectly correlated and signals are
arbitrarily precise, the intermediary's profits from the first-best policy
are nil. Under these conditions, an intermediary with commitment (or even an
intermediary who sells its own products) would not monetize their
information by selling it to consumers\textbf{.}

It is beyond the scope of this article to characterize the optimal
commitment policy for any initial data structure, but the data externality
clearly remains a key driver of the equilibrium allocation of information
even under stronger commitment assumptions.

\section{Conclusion\label{con}}

We have explored the trading of information between data intermediaries with
market power and multiple consumers with correlated preferences. The data
externality that we have uncovered strongly suggests that levels of
compensation close to zero can induce an individual consumer in a large
market to relinquish precise information about her preferences. This finding
holds even if the consumer's data are later sold to a firm that seeks to
extract their surplus. Thus, giving consumers control rights over their data
(a pillar of privacy regulation such as the EU General Data Protection
Regulation or the California Privacy Rights Act) is insufficient to bring
about the efficient use of information.

Our results regarding the aggregation of consumer information further
suggest that privacy regulations must move away from concerns over
personalized prices at the individual level. Most often, firms do not set
prices in response to individual-level characteristics. Instead,
segmentation of consumers occurs at the group level (e.g., as in the case of
Uber) or at the temporal and spatial levels (e.g., as in the case of Staples
and Amazon). Thus, our analysis points to the significant welfare effects of
group-level and market-level, dynamic prices that react in real time to
changes in demand.

A possible mitigator of the consequences of data externalities---echoed in 
\cite{powe18}---consists of facilitating the formation of consumer groups or
unions to internalize the data externality when bargaining with powerful
intermediaries, such as large online platforms.\footnote{%
This result echoes the claim in \cite{zubo19} that \textquotedblleft privacy
is a public issue.\textquotedblright} A different regulatory solution is
based on \emph{privacy managers}, such as internet browsers with
heterogeneous privacy settings that compete for consumers' default choice.
Yet another solution---suggested by \cite{romer19}---consists of making the
data outflow costly for the intermediary by, for example, taxing targeted
advertising. In our model, taxing the data outflow will limit efficient and
inefficient intermediation alike but will affect the intermediary's choice
of data policy under the assumptions of Section \ref{iwc}.

Finally, our data intermediary collected and redistributed the consumer data
but played no role in the interaction between the consumers and the
producer. In contrast, a consumer can often access a given producer only
through a data platform.\footnote{%
Product data platforms, such as Amazon, Uber and Lyft, acquire individual
data from the consumer through the consumers' purchase of services and
products. Social data platforms, such as Google and Facebook, offer data
services to individual users and sell the information to third parties, who
mostly purchase the information in the form of targeted advertising space.
In terms of our model, a product data platform combines the roles of data
intermediation and product pricing.} Many platforms can then be thought of
as auctioning access to the consumer. The data platform provides the bidding
producers with additional information that they can use to tailor their
interactions with consumers. Social data platforms thus trade individual
consumer information for services rather than money. In these markets, the
data externality manifests itself in the quality of the services offered and
in the extent of the consumers' engagement.\newpage

\section{Appendix A}

\begin{proof}[Proof of Lemma \protect\ref{tan}]
Under an arbitrary data-inflow policy $X$, each consumer $i$ observes a
noisy signal $S_{i}$ of her own willingness to pay and sends a potentially
noisier signal $X_{i}$ to the intermediary.\footnote{%
Under complete data sharing, for example, the consumer either reports $%
X_{i}=S_{i}$ or refuses to participate, so that $X_{i}$ has infinite
variance (or the corresponding $\sigma $-algebra is the empty set).}
Consumer $i$ observes both $S_{i}$ and $X_{i}$. Given the data inflow $X$,
the intermediary chooses an outflow policy, namely, the signal $%
Y_{0}=Y_{0}(X)$ sent to the producer and the signal $Y_{i}=Y_{i}(X)$ sent to
each consumer $i$. The intermediary then chooses a policy $Y$ that maximizes
the producer's ex ante expected payoff, which it fully extracts through the
fee $m_{0}$. We let the intermediary selects their favorite equilibrium in
the ensuing game.

For any outflow policy $Y=(Y_{0},Y_{i})$, denote an induced signaling
equilibrium as $\bar{\gamma}=(\bar{q}_{i},\bar{p})$, where $\bar{p}%
:Y_{0}\rightarrow R^{+}$ is the pricing strategy of the producer and $\bar{q}%
_{i}:Y_{i}\times S_{i}\times X_{i}\times R^{+}\rightarrow R^{+}$ is the
demand function of consumer $i$. We first argue that there exists an
equilibrium $\gamma ^{\ast }$ under the outflow policy $(\bar{p}\circ
Y_{0},(Y_{i},\bar{p}\circ Y_{0}))$ that yields a weakly higher ex ante
payoff for the producer. In this new outflow policy, instead of revealing $%
Y_{0}$ to the producer, the intermediary directly recommends the price $\bar{%
p}(Y)$ which coincides with the equilibrium pricing strategy in $\bar{\gamma}
$, and reveals to consumer $i$ both $Y_{i}$ and the price recommendation.

On the equilibrium path of $\bar{\gamma}$, consumer $i$ updates her
posterior $\mu (Y_{i},S_{i},\bar{p}_{i}(Y))$ using $Y_{i}$, her own private
signal $S_{i}$, the report $X_{i}$, and the observed price $p_{i}$. We
denote the consumer's demand as a function of her posterior beliefs and the
price as 
\begin{equation*}
q_{i}(\mu (Y_{i},S_{i},X_{i},p_{i}),p_{i}).
\end{equation*}%
The ex ante profit of the producer from consumer $i$ is given by 
\begin{equation*}
\mathbb{E}\left[ \bar{p}_{i}q_{i}(\mu (Y_{i},S_{i},X_{i},\bar{p}_{i}),\bar{p}%
_{i})\right] .
\end{equation*}

Now consider the new outflow policy $(\bar{p}\circ Y_{0},(Y_{i},\bar{p}\circ
Y_{0}))$. Under this policy, there exists an equilibrium where consumer $i$
forms her demand using the data outflow $(Y_{i},p^{\ast }\circ Y_{0})$ from
the intermediary as well as her own signal $S_{i}$ and the data inflow $%
X_{i} $. Because consumer $i$ knows everything that the producer knows, the
price charged by the producer no longer influences the consumer's posterior,
which therefore coincides with the on-path beliefs in the original
equilibrium $\bar{\gamma}$, i.e.,%
\begin{equation*}
\mu (Y_{i},S_{i},X_{i},\bar{p}_{i}(Y_{0})).
\end{equation*}%
Knowing this, the producer maximizes his ex ante payoff by choosing a
pricing strategy $\hat{p}(\cdot )$ as a function of his signal $\bar{p}\circ
Y_{0}$. Thus the producer's equilibrium profit is given by 
\begin{equation*}
\max_{\hat{p}}\hat{p}\big(\bar{p}\circ Y_{0}\big)q_{i}\Big(\mu
(Y_{i},S_{i},X_{i},\bar{p}_{i}(Y_{0})),\hat{p}\big(\bar{p}\circ Y_{0}\big)%
\Big).
\end{equation*}%
Clearly \textquotedblleft following the intermediary's
recommendation,\textquotedblright\ i.e., setting $\hat{p}(p)=p$ is a
feasible strategy that yields the same payoff as in the old equilibrium $%
\bar{\gamma}$. Consequently, the producer's equilibrium payoff is weakly
higher than in $\bar{\gamma}$.
\end{proof}

\begin{proof}[Proof of Proposition \protect\ref{impact}]
For any offered price $p_{i}$, consumer $i$ demands the quantity%
\begin{equation*}
q_{i}=\mathbb{E}[w_{i}|(S_{i},Y_{i})]-p_{i}.
\end{equation*}%
The producer finds it optimal to set the following price 
\begin{equation*}
p_{i}=\frac{\mathbb{E}[w_{i}|Y_{0}]}{2}.
\end{equation*}%
Recall that the consumer always has superior information so that $Y_{0}$ is
measurable with respect to $Y_{i}$. The profit of the producer is given by 
\begin{align}
\Pi _{i}((S_{i},Y_{i}),Y_{0})& =\mathbb{E}\left[ \frac{\mathbb{E}%
[w_{i}|Y_{0}]}{2}\left( \mathbb{E}[w_{i}|(S_{i},Y_{i})]-\frac{\mathbb{E}%
[w_{i}|Y_{0}]}{2}\right) \right]   \notag \\
& =\frac{\mathbb{E}\left[ (\mathbb{E}[w_{i}|Y_{0}])^{2}\right] }{4}=\frac{%
\limfunc{var}[\mathbb{E}[w_{i}|Y_{0}]]+\mathbb{E}[w_{i}]^{2}}{4}\\
&=\frac{1}{4}%
G\left( Y_{0}\right) +  \Pi_{i}\left(
S_{i},\varnothing \right) ,  \notag
\end{align}%
where the outside expectation represents integration over the whole
probability space. The expected consumer surplus is given by 
\begin{eqnarray}
U_{i}((S_{i},Y_{i}),Y_{0}) &=&\mathbb{E}\left[ \left( w_{i}-\frac{\mathbb{E}%
[w_{i}|Y_{0}]}{2}\right) \left( \mathbb{E}[w_{i}|(S_{i},Y_{i})]-\frac{%
\mathbb{E}[w_{i}|Y_{0}]}{2}\right) \right]   \notag \\
&&-\frac{1}{2}\mathbb{E}\left[ \left( \mathbb{E}[w_{i}|(S_{i},Y_{i})]-\frac{%
\mathbb{E}[w_{i}|Y_{0}]}{2}\right) ^{2}\right]  \\
&=&\frac{1}{2}\mathbb{E}[(\mathbb{E}[w_{i}|(S_{i},Y_{i})])^{2}-\frac{3}{4}(%
\mathbb{E}[w_{i}|Y_{0}])^{2}],  \notag \\
&=&\frac{1}{2}\left( G\left( (S_{i},Y_{i})\right) -G\left( S_{i}\right)
\right) -\frac{3}{8}G\left( Y_{0}\right) + U_{i}\left(
S_{i},\varnothing \right)    \label{consurplus}
\end{eqnarray}%
Finally, the impact on total surplus is given by the sum of the two effects: 
\begin{equation*}
W_{i}\left( (S_{i},Y_{i}),Y_{0}\right) -W_{i}\left( S_{i},\varnothing
\right) =\frac{1}{2}\left( G\left( (S_{i},Y_{i})\right) -G\left(
S_{i}\right) \right) -\frac{1}{8}G\left( Y_{0}\right) ,
\end{equation*}%
which completes the proof.
\end{proof}

\begin{proof}[Proof of Proposition \protect\ref{dop}]
By Lemma \ref{tan}, it is without loss of generality to assume the producer
receives a signal $Y,$ and the consumer receives a signal $(Y_{i},Y)$. Thus,
we can focus on equilibria where prices have no signaling effect. These
equilibria coincide with those described in Proposition \ref{impact}. As we
have shown there, the profit of the producer is: 
\begin{equation*}
\mathbb{E}\left[ \frac{\mathbb{E}[w_{i}|Y]}{2}\left( \mathbb{E}[w_{i}|Y\cup
Y_{i},S_{i},X_{i}]-\frac{\mathbb{E}[w_{i}|Y]}{2}\right) \right] =\frac{%
\mathbb{E}\left[ (\mathbb{E}[w_{i}|Y])^{2}\right] }{4}=\frac{\limfunc{var}[%
\mathbb{E}[w_{i}|Y]]+\mathbb{E}[w_{i}]^{2}}{4}.
\end{equation*}%
Therefore it is optimal to maximize $\limfunc{var}[\mathbb{E}[w_{i}|Y]]$,
which is achieved by setting $Y=X$. Hence,$\ $the intermediary reveals all
information collected $\left( Y=X\right) $ both to the producer and to
consumer $i$.\bigskip
\end{proof}

\begin{proof}[Proof of Corollary \protect\ref{IbP}]
When fundamentals $w_{i}$ are perfectly correlated, 
\begin{align*}
\mathbb{E}[w_{i}|S]& =\mathbb{E}[\theta |S]=\mathbb{E}[\theta
|S_{1},...,S_{N}], \\
\mathbb{E}[w_{i}|S_{-i}]& =\mathbb{E}[\theta |S_{-i}], \\
\limfunc{var}[\mathbb{E}[\theta |S_{-i}]]& =\limfunc{var}[\mathbb{E}[\theta
|S_{1},...,S_{N-1}]].
\end{align*}%
Under our symmetry assumption, the variance of the posterior expectation of
the common willingness to pay $\limfunc{var}[\mathbb{E}[\theta |S_{1,..,N}]]$
can be written as a function of $N$. Now we argue that $\limfunc{var}[%
\mathbb{E}[\theta |S_{1,..,N}]]$ is increasing in $N$. We first define $%
g(S_{1,...,N-1})\triangleq \mathbb{E}[\theta |S_{1,..,N-1}]$. Then,
according to Lemma \ref{statis} below, we have 
\begin{align*}
\limfunc{var}[\mathbb{E}[\theta |S_{1,..,N}]]& =\max_{f\in L^{2}}~ \limfunc{var}[\theta ]-%
\mathbb{E}[(\theta -g(S_{1,...,N}))^{2}], \\
& \geq \max_{f\in L^{2}}~\limfunc{var}[\theta ]-\mathbb{E}[(\theta
-g(S_{1,...,N-1}))^{2}], \\
& =\limfunc{var}[\mathbb{E}[\theta |S_{1,..,N-1}]].
\end{align*}%
The sequence $\limfunc{var}[\mathbb{E}[\theta |S_{1,..,N}]]$ is increasing
and bounded. Therefore, it converges: 
\begin{equation*}
\lim_{N\rightarrow \infty }G(S)=\lim_{N\rightarrow \infty }G(S_{-i}),
\end{equation*}%
and intermediation is then profitable: 
\begin{equation*}
\lim_{N\rightarrow \infty }\frac{R(S)}{N}=\frac{1}{4}\lim_{N\rightarrow
\infty }G(S)>0.
\end{equation*}%
In the limit for $N\rightarrow \infty $, the data externality and the
consumer surplus are given by 
\begin{align*}
\lim_{N\rightarrow \infty }U_{i}(S,S)-U_{i}(S_{i},\varnothing )&
=\lim_{N\rightarrow \infty }\mathbb{E}[\frac{1}{8}(\mathbb{E}[w_{i}|S])^{2}-%
\frac{1}{2}(\mathbb{E}[w_{i}|S_{i}]-\mathbb{E}[w_{i}])^{2}-\frac{1}{8}%
\mathbb{E}[w_{i}]^{2}] \\
& =\lim_{N\rightarrow \infty }\frac{1}{8}\limfunc{var}[\mathbb{E}[w_{i}|S]]-%
\frac{1}{2}\limfunc{var}[\mathbb{E}[w_{i}|S_{i}]], \\
\lim_{N\rightarrow \infty }DE_{i}(S)& =\lim_{N\rightarrow \infty }\frac{1}{2}%
\limfunc{var}[\mathbb{E}[w_{i}|S]]-\frac{1}{2}\limfunc{var}[\mathbb{E}%
[w_{i}|S_{i}]]-\frac{3}{8}\limfunc{var}[\mathbb{E}[w_{i}|S_{-i}]] \\
& =\lim_{N\rightarrow \infty }\frac{1}{8}\limfunc{var}[\mathbb{E}[w_{i}|S]]-%
\frac{1}{2}\limfunc{var}[\mathbb{E}[w_{i}|S_{i}]].
\end{align*}%
Therefore, when the initial noise is sufficiently small (i.e., when $%
\limfunc{var}[\mathbb{E}[w_{i}|S_{i}]]$ is close to $\limfunc{var}[w_{i}]$),
the data externality is negative and data sharing hurts consumers.
\end{proof}

\begin{proof}[Proof of Corollary \protect\ref{EbU}]
Because $w_{i}$ is independent from the other consumers' signals, we have $%
\limfunc{var}[\mathbb{E}[w_{i}|S_{-i}]]=0$. Thus, intermediation is always
unprofitable, and the data externality is always positive, 
\begin{align*}
R(S)& =-\frac{N}{8}\limfunc{var}[\mathbb{E}[w_{i}|S]]<0, \\
DE_i(S)& =\frac{1}{2}(\limfunc{var}[\mathbb{E}[w_{i}|S]]-\limfunc{var}[\mathbb{%
E}[w_{i}|S_{i}]] )\geq 0.
\end{align*}%
Finally, for the results on consumer surplus, we turn to Lemma \ref{statis}.
In particular, we know 
\begin{align*}
\limfunc{var}[\mathbb{E}[w_{i}|S]]& =\limfunc{var}[w_{i}]-\mathbb{E}[(w_{i}-%
\mathbb{E}[w_{i}|S])^{2}], \\
& \geq \limfunc{var}[w_{i}]-\mathbb{E}[(w_{i}-(s_{i}-\frac{1}{N-1}\Sigma
_{j\neq i}s_{j}))^{2}], \\
& =\limfunc{var}[\theta _{i}]-\mathbb{E}[(\theta _{i}-(\theta
_{i}+\varepsilon -\frac{1}{N-1}\Sigma _{j\neq i}\theta _{j}-\varepsilon
))^{2}], \\
& =\limfunc{var}[\theta _{i}]-\mathbb{E}[(\frac{1}{N-1}\Sigma _{j\neq
i}\theta _{j})^{2}]=\frac{N-2}{N-1}\limfunc{var}[\theta _{i}]\rightarrow 
\limfunc{var}[w_{i}].
\end{align*}%
Thus we obtain 
\begin{equation*}
\lim_{N\rightarrow \infty }U_{i}(S,S)-U_{i}(S_{i},\varnothing )=\frac{1}{8}%
\limfunc{var}[w_{i}]-\frac{1}{2}\limfunc{var}[\mathbb{E}[w_{i}|S_{i}]].
\end{equation*}%
When $\sigma $ is sufficiently large, so that $\limfunc{var}[\mathbb{E}%
[w_{i}|S_{i}]]$ is close to $0$, intermediation increases consumer surplus.
\end{proof}

The proof of Proposition \ref{prof} follows from expressions (\ref{mistar})
and (\ref{revx}) in the text.\bigskip

\begin{proof}[Proof of Proposition \protect\ref{si}]
In the main text, the data inflow from consumer $i$ is given by $X_{i}=S_{i}$
(under complete sharing) and we compare it with $X_{i}^{\ast }=\delta \left(
S_{i}\right) $ (under anonymization). Note that $(S_{i},X_{i})_{i}$ in this
case is symmetrically distributed, i.e., its joint density is unchanged
under permutations of indices. Here we prove a slightly more general version
of the result by allowing an arbitrary information inflow $X$ such that $%
(S_{i},X_{i})_{i}$ is symmetrically distributed.\footnote{%
For example, in Section \ref{idd} $X_{i}$ might be a noisier signal of $%
S_{i} $.} We assume that apart from the private signal $S_i$ and information outflow $Y_i$ provided by the intermediary, consumer $i$ also observes her own data inflow $X_i$. This assumption is needed because information set $%
(S_{i},X_{i},X_{-i}^{\ast })$ and $(S_{i},X_{i},X^{\ast })$ are equivalent
by construction, but $(S_{i},X_{-i}^{\ast })$ and $(S_{i},X^{\ast })$ maybe
not. Note that when we restrict to the less general case (when $X_{i}=S_{i}$%
), the latter holds automatically, so we do not need this assumption.

For any fixed inflow policy $X$, we refer to $p_{-i}$ as the off-path price
charged to consumer $i$ when she does not accept the intermediary's
contract, and to $p_{i}$ as the on-path price charged to consumer $i$. Now
consider another inflow policy $X^{\ast }$ identical to $X$ up to a random
permutation of the consumers' identities. Under this scheme, we refer to $%
p_{-i}^{\ast }$ as the off-path price for consumer $i$, and to $p_{i}^{\ast
} $ as the on-path price for consumer $i$.

We first argue that $p_{-i}=p_{-i}^{\ast }$ for any realization of $W,S,X$.
To do so, let us calculate consumer $i$'s posterior about $W_{i}$ under each
inflow policy. Under the non-anonymized scheme, the posterior distribution
of consumer $i$'s willingness to pay is given by 
\begin{align*}
f_{i}(W_{i}& =w_{i}|S_{i}=s_{i},X_{i}=x_{i},X=x) \\
& =\frac{\int f(W_{i}=w_{i},W_{-i}=w_{-i}^{\prime
},S_{i}=s_{i},S_{-i}=s_{-i}^{\prime },X_{i}=x_{i},X_{-i}=x_{-i})\text{d}%
s_{-i}^{\prime }\text{d}w_{-i}^{\prime }}{\int f(W=w^{\prime
},S_{i}=s_{i},S_{-i}=s_{-i}^{\prime },X_{i}=x_{i},X_{-i}=x_{-i})\text{d}%
s_{-i}^{\prime }\text{d}w^{\prime }}.
\end{align*}

Recall from Proposition \ref{dop} that the intermediary's optimal data
outflow policy consists of revealing to the consumers all the available
information, even if the consumer refuses to participate. When the data is
anonymized, because consumer $i$ knows her own report $X_{i}$, the data
outflow reveals to her the vector of reports $X_{-i}$ without knowing who
generated each one. We now define $\delta \in S^{n-1}$ as permutation of
consumer indices. Consumer $i$'s posterior distribution over her willingness
to pay $w_{i}$ is now given by 
\begin{equation*}
f_{i}(W_{i}=w_{i}|S_{i}=s_{i},X_{i}=x_{i},X_{-i}^{\ast }=x_{-i}).
\end{equation*}%
For notational simplicity, we use $Pr$ to denote both probability and the proper marginal density. Then the posterior can be rewritten as
\begin{align*}
& \frac{\Pr (W_{i}=w_{i},S_{i}=s_{i},X_{i}=x_{i},X_{-i}^{\ast }=x_{-i})}{\Pr (S_{i}=s_{i},X_{i}=x_{i},X_{-i}^{\ast }=x_{-i})} \\
& =\frac{\Sigma _{\delta \in S^{n-1}}\Pr (\delta
,W_{i}=w_{i},S_{i}=s_{i},X_{i}=x_{i},X_{-i}=x_{\delta (-i)})}{\Sigma
_{\delta \in S^{n-1}}\Pr (\delta ,S_{i}=s_{i},X_{i}=x_{i},X_{-i}=x_{\delta
(-i)})} \\
& =\frac{\Sigma _{\delta \in S^{n-1}}\Pr (\delta )\Pr
(W_{i}=w_{i},S_{i}=s_{i},X_{i}=x_{i},X_{-i}=x_{\delta (-i)})}{\Sigma
_{\delta \in S^{n-1}}\Pr (\delta
)Pr(S_{i}=s_{i},X_{i}=x_{i},X_{-i}=x_{\delta (-i)})}.
\end{align*}%
Because of the symmetry assumption, we know that 
\begin{align*}
& \Pr (W_{i}=w_{i},S_{i}=s_{i},X_{i}=x_{i},X_{-i}=x_{\delta (-i)}) \\
& =\int f(W_{i}=w_{i},W_{-i}=w_{-i}^{\prime
},S_{i}=s_{i},S_{-i}=s_{-i}^{\prime },X_{i}=x_{i},X_{-i}=x_{\delta (-i)})%
\text{d}s_{-i}^{\prime }\text{d}w_{-i}^{\prime } \\
& =\int f(W_{i}=w_{i},W_{-i}=w_{\delta ^{-1}(-i)}^{\prime
},S_{i}=s_{i},S_{-i}=s_{\delta ^{-1}(-i)}^{\prime
},X_{i}=x_{i},X_{-i}=x_{-i})\text{d}s_{-i}^{\prime }\text{d}w_{-i}^{\prime }
\\
& =\int f(W_{i}=w_{i},W_{-i}=w_{\delta ^{-1}(-i)}^{\prime
},S_{i}=s_{i},S_{-i}=s_{\delta ^{-1}(-i)}^{\prime
},X_{i}=x_{i},X_{-i}=x_{-i})\text{d}s_{\delta ^{-1}(-i)}^{\prime }\text{d}%
w_{\delta ^{-1}(-i)}^{\prime } \\
& =\Pr (W_{i}=w_{i},S_{i}=s_{i},X_{i}=x_{i},X_{-i}=x_{-i}).
\end{align*}%
For the same reason, we also have 
\begin{equation*}
\Pr (S_{i}=s_{i},X_{i}=x_{i},X_{-i}=x_{\delta (-i)})=\Pr
(S_{i}=s_{i},X_{i}=x_{i},X_{-i}=x_{-i}).
\end{equation*}%
Thus the posterior of consumer $i$ can be simplified as: 
\begin{align*}
f_{i}(W_{i}& =w_{i}|S_{i}=s_{i},X_{i}=x_{i},X_{-i}^{\ast }=x_{-i}) \\
& =\frac{\Sigma _{\delta \in S^{n-1}}\Pr (\delta )\Pr (W_{i}=w_{i},S_{i}=s_{i},X_{i}=x_{i},X_{-i}=x_{-i})}{\Sigma _{\delta \in
S^{n-1}}\Pr (\delta )Pr(S_{i}=s_{i},X_{i}=x_{i},X_{-i}=x_{-i})} \\
& =\frac{\Sigma _{\delta \in S^{n-1}}\frac{1}{|S^{n-1}|}\Pr (W_{i}=w_{i},S_{i}=s_{i},X_{i}=x_{i},X_{-i}=x_{-i})}{\Sigma _{\delta \in
S^{n-1}}\frac{1}{|S^{n-1}|}\Pr (S_{i}=s_{i},X_{i}=x_{i},X_{-i}=x_{-i})} \\
& =f_{i}(W_{i}=w_{i}|S_{i}=s_{i},X_{i}=x_{i},X_{-i}=x_{-i}).
\end{align*}%
We have therefore proved that consumer $i$ has the same posterior about her
willingness to pay $w_{i}$ for any realization of $W,S,X$ irrespective of
whether the data are anonymized of not. Furthermore, this holds both on and
off the path of play.

Next, we show that the producer also has the same posterior about $W_{i}$
for any realization of $W,S,X$ when consumer $i$ refuses to report. Under
the non-anonymized scheme, the posterior density is given by: 
\begin{equation*}
f_{i}(W_{i}=w_{i}|X=x)=\frac{\int f(W_{i}=w_{i},W_{-i}=w_{-i}^{\prime
},S=s^{\prime },X_{i}=x_{i},X_{-i}=x_{-i})\text{d}s^{\prime }\text{d}%
w_{-i}^{\prime }}{\int f(W=w^{\prime },S=s^{\prime },X=x_{i},X_{-i}=x_{-i})%
\text{d}s^{\prime }\text{d}w^{\prime }}.
\end{equation*}%
Under the anonymized scheme, the posterior density is given by 
\begin{equation*}
f_{i}(W_{i}=w_{i}|X^{\ast }=x)=\frac{\Sigma _{\delta \in S^{n-1}}\Pr (\delta
)\Pr (W_{i}=w_{i},X=\delta (x))}{\Sigma _{\delta \in S^{n-1}}\Pr (\delta
)\Pr (X=\delta (x))}
\end{equation*}%
By the earlier argument, we can simplify it as follows: 
\begin{equation*}
\frac{\Sigma _{\delta \in S^{n-1}}\Pr (\delta )\Pr (W_{i}=w_{i},X=x)}{\Sigma
_{\delta \in S^{n-1}}\Pr (\delta )\Pr (X=x)}=f_{i}(W_{i}=w_{i}|X=x)
\end{equation*}%
Because the posteriors for both parties are the same for any realization, so
is the price, and hence the welfare impact of information

The profit of the intermediary from consumer $i$'s data under inflow policy $%
X$ is given by 
\begin{equation*}
R_{i}\left( X\right) =\Pi (X,X)-\Pi (S_{i},\varnothing
)-U_{i}((S_{i},X_{-i}),X_{-i})+U_{i}((S_{i},X),X).
\end{equation*}%
We have argued that consumer surplus off the path is the same: 
\begin{equation*}
U_{i}((S_{i},X_{-i}),X_{-i})=U_{i}((S_{i},X^{\ast }_{-i}),X_{-i}^{\ast }).
\end{equation*}%
We now turn to the last term---the impact on social welfare on the path of
play: 
\begin{align*}
& \Pi ((S_{i},X),X)+U_{i}((S_{i},X),X) \\
& =\frac{1}{2}\mathbb{E}[(\mathbb{E}[w_{i}|S_{i},X_{i},X]-\mathbb{E}%
[w_{i}|X])^{2}+\frac{1}{4}(\mathbb{E}[w_{i}|X])^{2}]+\frac{\limfunc{var}[%
\mathbb{E}[w_{i}|X]]}{4} \\
& =\frac{1}{2}\limfunc{var}[\mathbb{E}[w_{i}|S_{i},X_{i},X]]-\frac{1}{8}%
\limfunc{var}[\mathbb{E}[w_{i}|X]].
\end{align*}%
Recall that consumer $i$ has the same on path posterior under two different
scheme. Therefore, the difference in the intermediary's profits under the
two policies reduces to 
\begin{align*}
& \frac{1}{2}\limfunc{var}[\mathbb{E}[w_{i}|S_{i},X_{i},X]]-\frac{1}{8}%
\limfunc{var}[\mathbb{E}[w_{i}|X]]-\frac{1}{2}\limfunc{var}[\mathbb{E}%
[w_{i}|S_{i},X_{i},X^{\ast }]]+\frac{1}{8}\limfunc{var}[\mathbb{E}%
[w_{i}|X^{\ast }]] \\
& =-\frac{1}{8}\limfunc{var}[\mathbb{E}[w_{i}|X]]+\frac{1}{8}\limfunc{var}[%
\mathbb{E}[w_{i}|X^{\ast }]]\leq 0.
\end{align*}%
Therefore, anonymization is more profitable than complete sharing, and
strictly so whenever anonymization makes the estimation less precise.
\end{proof}

In the remainder of the Appendix, we often make use of the following
classical result in statistics, which we state as a lemma without
proof---the result is an immediate consequence of the fact that $\mathbb{E}%
[X|Y]$ is the projection of $X$ on $\mathcal{F}(Y)$ in $L^{2}$ space.

\begin{lemma}
\label{statis}Let $W$ and $Y$ be two random variables. Then it holds that 
\begin{equation*}
\limfunc{var}[\mathbb{E}[W|Y]]=\limfunc{var}[W]-\mathbb{E}[(W-\mathbb{E}
[W|Y])^{2}]\leq \limfunc{var}[W],
\end{equation*}
and 
\begin{equation*}
\text{ }\mathbb{E}[(W-\mathbb{E}[W|Y])^{2}]\leq \mathbb{E}
[(W-f(Y))^{2}],\quad \forall f\in L^{2}.
\end{equation*}
\end{lemma}

To prove Proposition \ref{prof_anonym}, we first state a basic property of
anonymized data sharing in our symmetric environment.

\begin{lemma}
\label{lemm_sym_est} When the data is anonymized, the following holds: 
\begin{equation*}
\mathbb{E}[w_{i}|A]=\mathbb{E}[w_{j}|A].
\end{equation*}
\end{lemma}

\begin{proof}[Proof of Lemma \protect\ref{lemm_sym_est}]
Denote the joint distribution of $W$ and $S$ as $f(W=w,S=s)$ and the
posterior of $W_{i}$ after observing $A$ as $f(W_{i}=w|A)$. Denote the
permutation in $S^{N}$ as $\nu $ and especially the swapping between $i$ and 
$j$ as $\nu _{ij}$. For notational simplicity, we use $\Pr $ to denote both
probability and the proper marginal density. 
\begin{align*}
f_{i}(W_{i}& =w_{i}|A=s)=\frac{\Pr (W_{i}=w_{i},A=s)}{\Pr (A=s)}=\frac{%
\Sigma _{\nu \in S^{N}}\Pr (\nu )\Pr (W_{i}=w_{i},S_{\nu }=s))}{\Pr (A=s)} \\
& =\frac{\Sigma _{\nu \in S^{N}}\frac{1}{|S^{N}|}\int
f(W_{i}=w_{i},W_{j}=w_{j},W_{-ij}=w_{-ij},S_{\nu }=s)\text{d}w_{j}\text{d}%
w_{-ij}}{\Pr (A=s)}.
\end{align*}%
Because $f$ is unchanged under permutation, we can apply the following
transformation: 
\begin{align*}
& f_{i}(W_{i}=w_{i}|A=s)=\frac{\Sigma _{\nu \in S^{N}}\frac{1}{|S^{N}|}\int
f(W_{j}=w_{i},W_{i}=w_{j},W_{-ij}=w_{-ij},S_{\nu _{ij}\circ \nu }=s)\text{d}%
w_{j}\text{d}w_{-ij}}{\Pr (A=s)}, \\
& =\frac{\Sigma _{\nu _{ij}\circ \nu \in S^{N}}\frac{1}{|S^{N}|}\int
f(W_{j}=w_{i},W_{i}=w_{i}^{\prime },W_{-ij}=w_{-ij},S_{\nu _{ij}\circ \nu
}=s)\text{d}w_{i}^{\prime }\text{d}w_{-ij}}{\Pr (A=s)}%
=f_{j}(W_{j}=w_{i}|A=s).
\end{align*}%
Because the posterior distribution is the same, so is the conditional
expectation because 
\begin{equation*}
\mathbb{E}[w_{i}|A]=\int w_{i}f_{i}(W_{i}=w_{i}|A)\text{d}w_{i},
\end{equation*}%
which completes the proof.\bigskip
\end{proof}

\begin{proof}[Proof of Proposition \protect\ref{prof_anonym}]
Combining Lemmas \ref{statis} and \ref{lemm_sym_est}, we obtain 
\begin{align*}
\mathbb{E}[w_{i}|A]& =\mathbb{E}[\frac{1}{N}\Sigma _{i}w_{i}|A]=\mathbb{E}%
[\theta +\frac{1}{N}\Sigma _{i}\theta _{i}|A]; \\
G(A)& =\limfunc{var}[\mathbb{E}[\theta +\frac{1}{N}\Sigma _{i}\theta
_{i}|A]]=\limfunc{var}[\theta +\frac{1}{N}\Sigma _{i}\theta _{i}]-\mathbb{E}%
[(\theta +\frac{1}{N}\Sigma _{i}\theta _{i}-\mathbb{E}[\theta +\frac{1}{N}%
\Sigma _{i}\theta _{i}|A])^{2}].
\end{align*}%
We can simplify the last term as follows: 
\begin{align*}
& \mathbb{E}[(\theta +\frac{1}{N}\Sigma _{i}\theta _{i}-\mathbb{E}[\theta +%
\frac{1}{N}\Sigma _{i}\theta _{i}|A])^{2}] \\
& =\mathbb{E}[(\theta -\mathbb{E}[\theta |A])^{2}+\frac{1}{N^{2}}(\Sigma
_{i}\theta _{i}-\Sigma _{i}\mathbb{E}[\theta _{i}|A])^{2}-\frac{2}{N}(\theta
-\mathbb{E}[\theta |A])(\Sigma _{i}\theta _{i}-\Sigma _{i}\mathbb{E}[\theta
_{i}|A])] \\
& \geq \mathbb{E}[(\theta -\mathbb{E}[\theta |A])^{2}]-\frac{2}{N}\sqrt{%
\limfunc{var}[\theta -\mathbb{E}[\theta |A]]\limfunc{var}[\Sigma _{i}\theta
_{i}-\Sigma _{i}\mathbb{E}[\theta _{i}|A]]} \\
& \geq \mathbb{E}[(\theta -\mathbb{E}[\theta |A])^{2}]-\frac{2}{N}\sqrt{N%
\limfunc{var}[\theta ]\limfunc{var}[\theta _{i}]},
\end{align*}%
where the last inequality comes from Lemma \ref{statis}. The intermediary's
profit can be written as 
\begin{align*}
R& =3G(A_{-i})-G(A), \\
& =3\limfunc{var}[\mathbb{E}[\theta |A_{-i}]]-\limfunc{var}[\theta ]-\frac{1%
}{N}\limfunc{var}[\theta _{i}]+\mathbb{E}[(\theta +\frac{1}{N}\Sigma
_{i}\theta _{i}-\mathbb{E}[\theta +\frac{1}{N}\Sigma _{i}\theta
_{i}|A])^{2}], \\
& \geq 3\limfunc{var}[\mathbb{E}[\theta |A_{-i}]]-\limfunc{var}[\theta ]-%
\frac{1}{N}\limfunc{var}[\theta _{i}]+\mathbb{E}[(\theta -\mathbb{E}[\theta
|A])^{2}]-\frac{2}{N}\sqrt{N\limfunc{var}[\theta ]\limfunc{var}[\theta _{i}]}
\\
& =3\limfunc{var}[\mathbb{E}[\theta |A_{-i}]]-\limfunc{var}[\mathbb{E}%
[\theta |A]]-\frac{1}{N}\limfunc{var}[\theta _{i}]-\frac{2}{\sqrt{N}}\sqrt{%
\limfunc{var}[\theta ]\limfunc{var}[\theta _{i}]}.
\end{align*}%
Therefore, in the limit we have: 
\begin{equation*}
\lim_{N\rightarrow \infty }R=2\lim_{N\rightarrow \infty }\limfunc{var}[%
\mathbb{E}[\theta |A]]>0,
\end{equation*}%
which completes the proof.\bigskip
\end{proof}

\begin{proof}[Proof of Proposition \protect\ref{lm}]
We first prove that the total compensation is bounded from above, which
immediately implies that the individual compensation goes to $0$ as $%
N\rightarrow \infty $. From Lemma \ref{lemm_sym_est}, we know that 
\begin{align*}
G(A)& =\limfunc{var}[\mathbb{E}[w_{i}|A]]=\limfunc{var}[\mathbb{E}[\Sigma
_{i}\frac{w_{i}}{N}|A]], \\
& \leq \limfunc{var}[\Sigma _{i}\frac{w_{i}}{N}]=\limfunc{var}[\theta ]+%
\frac{\limfunc{var}[\theta _{i}]+\limfunc{var}[\varepsilon _{i}]}{N}.
\end{align*}%
On the other hand, we also know 
\begin{equation*}
G(A_{-i})=\limfunc{var}[\mathbb{E}[\theta |A_{-i}]]=\limfunc{var}[\theta ]-%
\mathbb{E}[(\theta -\mathbb{E}[\theta |A_{-i}])^{2}].
\end{equation*}%
Because the conditional expectation is the best $L^{2}$ approximation, we
know it leads to a smaller error than the \textquotedblleft sample average
estimator,\textquotedblright 
\begin{equation*}
\mathbb{E}\left[ (\theta -\mathbb{E}[\theta |A_{-i}])^{2}\right] \leq 
\mathbb{E}\left[ \theta -\frac{1}{N-1}\Sigma _{j\neq i}(\theta +\theta
_{j}+\varepsilon _{j})^{2}\right] =\frac{1}{N-1}(\limfunc{var}[\theta _{i}]+%
\limfunc{var}[\varepsilon _{j}]).
\end{equation*}%
Therefore, we have: 
\begin{align*}
N(G(A)-G(A_{-i}))& \leq N\Big(\limfunc{var}[\theta ]+\frac{\limfunc{var}%
[\theta _{i}]+\limfunc{var}[\varepsilon _{i}]}{N}-\limfunc{var}[\theta ]+%
\frac{1}{N-1}(\limfunc{var}[\theta _{i}]+\limfunc{var}[\varepsilon _{j}])%
\Big), \\
& =N\left( \frac{\limfunc{var}[\theta _{i}]+\limfunc{var}[\varepsilon _{i}]}{%
N}+\frac{1}{N-1}(\limfunc{var}[\theta _{i}]+\limfunc{var}[\varepsilon
_{i}])\right) \\
& \leq 3(\limfunc{var}[\theta _{i}]+\limfunc{var}[\varepsilon _{i}]).
\end{align*}%
The total consumer compensation is then given by 
\begin{equation*}
\frac{3N}{8}(G(A)-G(A_{-i}))\leq \frac{9}{8}(\limfunc{var}[\theta _{i}]+%
\limfunc{var}[\varepsilon _{i}]).
\end{equation*}

Finally, the intermediary's profit is growing linearly in $N$ because 
\begin{align*}
R(S)& =\frac{N}{4}G(A)-\frac{3N}{8}(G(A)-G(A_{-i})), \\
\lim_{N\rightarrow \infty }\frac{R(S)}{N}& =\frac{1}{4}\lim_{N\rightarrow
\infty }G(A),
\end{align*}%
which completes the proof.\bigskip
\end{proof}

\begin{proof}[Proof of Proposition \protect\ref{limrev}]
When data is not anonymized we have: 
\begin{equation*}
G(S)-G(S_{-i})=\limfunc{var}[\mathbb{E}[\theta +\theta _{i}|S]]-\limfunc{var}%
[\mathbb{E}[\theta |S_{-i}]].
\end{equation*}%
Because of symmetry, we have 
\begin{equation*}
\limfunc{cov}[\mathbb{E}[\theta |S],\mathbb{E}[\theta _{i}|S]]=\limfunc{cov}[%
\mathbb{E}[\theta |S],\mathbb{E}[\theta _{j}|S]]=\limfunc{cov}[\mathbb{E}%
[\theta |S],\Sigma _{j=1}^{N}\mathbb{E}[\theta _{j}/N\left\vert S\right. ]].
\end{equation*}%
Because the correlation coefficient is always greater than $-1$, we obtain 
\begin{align*}
\limfunc{cov}[\mathbb{E}[\theta |S],\Sigma _{j=1}^{N}\mathbb{E}[\theta
_{j}/N\left\vert S\right. ]]& \geq -\sqrt{\limfunc{var}[\theta ]\limfunc{var}%
[\Sigma _{j=1}^{N}\mathbb{E}[\theta _{j}/N\left\vert S\right. ]]}, \\
& \geq -\sqrt{\limfunc{var}[\theta ]\limfunc{var}[\Sigma _{j=1}^{N}\theta
_{j}/N]}.
\end{align*}%
Therefore, according to Lemma \ref{statis} we have: 
\begin{align*}
G(S)-G(S_{-i})& =\limfunc{var}[\mathbb{E}[\theta |S]]+2\limfunc{cov}[\mathbb{%
E}[\theta |S],\mathbb{E}[\theta _{i}|S]]+\limfunc{var}[\mathbb{E}[\theta
_{i}|S]]-\limfunc{var}[\mathbb{E}[\theta |S_{-i}]] \\
& \geq \limfunc{var}[\mathbb{E}[\theta |S]]-2\frac{1}{\sqrt{N}}\sqrt{%
\limfunc{var}[\theta ]\limfunc{var}[\theta _{i}]}+\limfunc{var}[\mathbb{E}%
[\theta _{i}|S]]-\limfunc{var}[\mathbb{E}[\theta |S_{-i}]],
\end{align*}%
and hence 
\begin{equation*}
\liminf_{N\rightarrow \infty }G(S)-G(S_{-i})\geq \limfunc{var}[\mathbb{E}%
[\theta _{i}|S]].
\end{equation*}%
The last term is strictly positive because 
\begin{align*}
\limfunc{var}[\mathbb{E}[\theta _{i}|S]]& =\limfunc{var}[\theta _{i}]-%
\mathbb{E}[(\theta _{i}-\mathbb{E}[\theta _{i}|S])^{2}] \\
& \geq \limfunc{var}[\theta _{i}]-\mathbb{E}[(\theta _{i}-\frac{\limfunc{var}%
[\theta _{i}]}{\limfunc{var}[\theta _{i}]+\limfunc{var}[\theta ]+\limfunc{var%
}[e]}S_{i})^{2}], \\
& =\limfunc{var}[\theta _{i}]-(\limfunc{var}[\theta _{i}]-\frac{\limfunc{var}%
^{2}[\theta _{i}]}{\limfunc{var}[\theta _{i}]+\limfunc{var}[\theta ]+%
\limfunc{var}[e]}), \\
& =\frac{\limfunc{var}^{2}[\theta _{i}]}{\limfunc{var}[\theta _{i}]+\limfunc{%
var}[\theta ]+\limfunc{var}[e]}>0,
\end{align*}%
where the first inequality again uses Lemma \ref{statis}.
\end{proof}

\begin{proof}[Proof of Proposition \protect\ref{doc}]
In the standard \textquotedblleft divide and conquer\textquotedblright\
scheme, the compensation for the $i$-th consumer is the marginal loss of
revealing her information given that $i-1$ consumers reveal their signals: 
\begin{equation*}
\frac{3}{8}G(S_{1,...,i})-\frac{3}{8}G(S_{1,...,i-1}).
\end{equation*}%
Because in general we do not know whether this marginal loss is decreasing
in $i$, we consider the following revised version of divide and conquer,
where consumer $i$ receives 
\begin{equation*}
m_{i}=\max_{k\geq i}\frac{3}{8}G(S_{1,...,k})-\frac{3}{8}G(S_{1,...,k-1}).
\end{equation*}%
Under this payment scheme, it is a dominant strategy for consumer $1$ to
accept the offer. Moreover, it is optimal for consumer $i$ to accept the
offer, given that the first $i-1$ consumers accept. Using an identical proof
to Proposition \ref{lm}, we obtain 
\begin{align*}
\frac{3}{8}G(S_{1,...,i})-\frac{3}{8}G(S_{1,...,i-1})& \leq \frac{3}{8}(%
\frac{1}{i}+\frac{1}{i-1})(\limfunc{var}[\theta _{i}]+\limfunc{var}%
[\varepsilon _{i}]), \\
& \leq \frac{3}{4}\frac{1}{i-1}(\limfunc{var}[\theta _{i}]+\limfunc{var}%
[\varepsilon _{i}]).
\end{align*}%
Therefore, we obtain an upper bound on the compensation paid to consumer $i$
: 
\begin{equation*}
m_{i}\leq \max_{k\geq i}\frac{3}{4}\frac{1}{k-1}(\limfunc{var}[\theta _{i}]+%
\limfunc{var}[\varepsilon _{i}])=\frac{3}{4}\frac{1}{i-1}(\limfunc{var}%
[\theta _{i}]+\limfunc{var}[\varepsilon _{i}]).
\end{equation*}%
Finally, because we have 
\begin{align*}
\Sigma _{i}\frac{3}{8}(G(S_{1,...,i})-G(S_{1,...,i-1}))& \leq \frac{3}{4}%
(1+\Sigma _{i=3}^{N}\frac{1}{i-1})(\limfunc{var}[\theta _{i}]+\limfunc{var}%
[\varepsilon _{i}]), \\
& \leq \frac{3}{4}(1+\log N)(\limfunc{var}[\theta _{i}]+\limfunc{var}%
[\varepsilon _{i}]),
\end{align*}%
the total compensation grows at a speed less than $\log N$.
\end{proof}

\begin{proof}[Proof of Proposition \protect\ref{ga}]
The proof of this proposition is similar to that of Proposition \ref{si}. By
Lemma \ref{tan}, we know that the intermediary will transmit whatever
information it collected to all consumers. By homogeneity, we know the
consumer's posterior about their own fundamental $w_{ij}$ is the same
whether the signals are anonymized or not, and the producer's posterior
about any deviating consumer's fundamental is also the same under the two
schemes.

Denote the broker's revenue under the non-anonymized and anonymized scheme
as $R(X)$ and $R(X^{\ast })$, respectively. It then holds that%
\begin{align*}
R_{i}\left( X\right) & =\Pi (X,X)-\Pi (S_{i},\varnothing
)-U_{i}((S_{i},X_{-i}),X_{-i})+U_{i}((S_{i},X),X), \\
R_{i}\left( X^{\ast }\right) & =\Pi (X^{\ast },X^{\ast })-\Pi
(S_{i},\varnothing )-U_{i}((S_{i},X_{-i}^{\ast }),X_{-i}^{\ast
})+U_{i}((S_{i},X^{\ast }),X^{\ast }).
\end{align*}%
Our analysis in previous paragraph implies 
\begin{equation*}
U_{i}((S_{i},X_{-i}),X_{-i})=U_{i}((S_{i},X_{-i}^{\ast }),X_{-i}^{\ast }).
\end{equation*}%
Therefore the intermediary prefers anonymization if and only if 
\begin{equation*}
R_{i}(X^{\ast })-R_{i}(X)=W((S_{i},X^{\ast }),X^{\ast })-W((S_{i},X),X)\geq
0,
\end{equation*}%
which completes the proof.
\end{proof}

\begin{proof}[Proof of Proposition \protect\ref{gs}]
We first consider the case where the intermediary anonymizes all data,
including the group identities. Similar to the result in Lemma \ref%
{lemm_sym_est} (A.2 in the latest draft), we know that the producer offers one price to all consumers
on the path of play, 
\begin{equation*}
	\mathbb{E}[w_{ij}|A]=\mathbb{E}[w_{i^{\prime }j^{\prime }}|A].
\end{equation*}%
Denoting $N=\Sigma _{j}N_{j}$, we have
\begin{align*}
	\mathbb{E}[w_{i^{\prime }j^{\prime }}|A]& =\frac{1}{\Sigma _{j}N_{j}}\Sigma
	_{j}\Sigma _{i}\mathbb{E}[w_{ij}|A]=\frac{1}{\Sigma _{j}N_{j}}\Sigma
	_{j}\Sigma _{i}\mathbb{E}[\theta _{j}+\theta _{ij}|A] \\
	& =\Sigma _{j}\frac{N_{j}}{N}\mathbb{E}[\theta _{j}|A] + \frac{1}{N}\mathbb{E}[\Sigma_{ji} \theta_{ij}|A] .
\end{align*}%
Therefore we obtain an upper bound on the revenue per capita 
\begin{align*}
	\frac{R(A)}{N}& =\frac{3}{8}G(A_{-ij})-\frac{1}{8}G(A)=\frac{3}{8}[\mathbb{E}%
	[w_{ij}]|A_{-ij}]-\frac{1}{8}\limfunc{var}[\mathbb{E}[w_{ij}]|A] \\
	& \leq \frac{1}{4}\limfunc{var}[\mathbb{E}[w_{ij}]|A]=\frac{1}{4} \limfunc{var}[\Sigma _{j}\frac{N_{j}}{N}\mathbb{E}[\theta _{j}|A] + \frac{1}{N}\mathbb{E}[\Sigma_{ji} \theta_{ij}|A]] \\
	& \leq \frac{%
		1}{4}\frac{1}{N^{2}}\Sigma N_{j}^{2}\limfunc{var}[\theta _{j}] + \frac{1}{4N} \limfunc{var}[ \theta_{ij}].
\end{align*}
Next, consider the case where the intermediary reveals the group identity.
Instead of $A$ we use $A^{g}$ to denote the information that the producer
receives. By an argument similar to the proof of Lemma \ref{lemm_sym_est} (A.2 in the latest draft),
we know that (on path) the producer offers one price to all consumers in
each group: 
\begin{align*}
	\mathbb{E}[w_{ij}|A^{g}]& =\mathbb{E}[w_{i^{\prime }j}|A^{g}] \\
	& =\frac{1}{N_{j}}\Sigma _{i^{\prime }=1}^{N_{j}}\mathbb{E}[w_{i^{\prime
		}j}|A^{g}]=\frac{1}{N_{j}}\Sigma _{i^{\prime }=1}^{N_{j}}\mathbb{E}%
	[w_{i^{\prime }j}|A^{g}]=\mathbb{E}[\theta _{j}+\frac{1}{N_{j}}\Sigma
	_{i^{\prime }=1}^{N_{j}}\theta _{i^{\prime }j}|A^{g}].
\end{align*}%
When consumer $ij$ rejects the offer, the intermediary will know the group
identity of this deviating consumer and use all the available data to
estimate the demand: 
\begin{equation*}
	\mathbb{E}[w_{ij}|A_{-ij}^{g}]=\mathbb{E}[\theta _{j}+\theta
	_{ij}|A_{-ij}^{g}]=\mathbb{E}[\theta _{j}|A_{-ij}^{g}].
\end{equation*}%
The revenue that the intermediary obtains from consumer $ij$'s data is then
given by 
\begin{align*}
	& \frac{3}{8}\limfunc{var}[\mathbb{E}[w_{ij}]|A_{-ij}^{g}]-\frac{1}{8}%
	\limfunc{var}[\mathbb{E}[w_{ij}]|A^{g}], \\
	& =\frac{3}{8}\limfunc{var}[\mathbb{E}[\theta _{j}|A_{-ij}^{g}]]-\frac{1}{8}%
	\limfunc{var}[\mathbb{E}[\theta _{j}+\frac{1}{N_{j}}\Sigma _{i^{\prime
		}=1}^{N_{j}}\theta _{i^{\prime }j}|A^{g}]], \\
	& \geq \frac{3}{8}\limfunc{var}[\mathbb{E}[\theta _{j}|A_{-ij}^{g}]]-\frac{1%
	}{8}\limfunc{var}[\theta _{j}]-\frac{1}{8N_{j}}\limfunc{var}[\theta _{ij}]-%
	\frac{1}{4}\sqrt{\frac{1}{N_{j}}\limfunc{var}[\theta _{j}]}\sqrt{\limfunc{var%
		}[\theta _{ij}]}.
\end{align*}

From Lemma \ref{statis} (A.1 in the latest draft), we know 
\begin{align*}
	\mathbb{E}[(\theta _{j}-\mathbb{E}[\theta _{j}|A_{-ij}^{g}])^{2}]& \leq 
	\mathbb{E}[(\theta _{j}-\frac{1}{N_{j}-1}\Sigma _{i^{\prime }\neq
		i}s_{i^{\prime }j})^{2}], \\
	& =\mathbb{E}[(\frac{1}{N_{j}-1}\Sigma _{i^{\prime }\neq i}-\theta _{i'j} - \varepsilon_{i'j})^{2}]=\frac{\limfunc{var}[\theta
		_{ij}] + \limfunc{var}[\varepsilon_{ij}]}{N_{j}-1}; \\
	\limfunc{var}[\mathbb{E}[\theta _{j}|A_{-ij}^{g}]]& =\limfunc{var}[\theta
	_{j}]-\mathbb{E}[(\theta _{j}-\mathbb{E}[\theta _{j}|A_{-ij}^{g}])^{2}], \\
	& \geq \limfunc{var}[\theta _{j}]-\frac{\limfunc{var}[\theta
		_{ij}] + \limfunc{var}[\varepsilon_{ij}]}{N_{j}-1}.
\end{align*}%
Thus we obtain a lower bound on the revenue from consumer $ij$: 
\begin{equation*}
	\frac{1}{4}\limfunc{var}[\theta _{j}]-\frac{3}{8}\frac{\limfunc{%
			var}[\theta _{ij}] + \limfunc{%
			var}[\varepsilon _{ij}]}{N_{j}-1}-\frac{1}{8N_{j}}\limfunc{var}[\theta _{ij}]-\frac{1}{4}%
	\sqrt{\frac{1}{N_{j}}\limfunc{var}[\theta _{j}]}\sqrt{\limfunc{var}[\theta
		_{ij}]}.
\end{equation*}

Finally we can compute the difference in the revenues 
\begin{align*}
	R(A)-R(A^{g})& \leq \Sigma _{j}\frac{N_{j}^{2}}{4N}\limfunc{var}[\theta _{j}] + \frac{1}{4} \limfunc{var}[\theta_{ij}]  + \Sigma_j \frac{3}{8}%
	\frac{N_{j} (\limfunc{var}[\theta _{ij}] + \limfunc{var}[\varepsilon_{ij}] )}{N_{j}-1}
	\\
	& -\Sigma _{j}\left( \frac{N_{j}}{4}\limfunc{var}[\theta _{j}] - \frac{1}{8}\limfunc{var}%
	[\theta _{ij}] -\frac{\sqrt{N_{j}}}{4}\sqrt{\limfunc{var}[\theta _{j}]}\sqrt{%
		\limfunc{var}[\theta _{ij}]}\right) .
\end{align*}%
As long as $N_{j}<kN$ where $k<1$, we know that 
\begin{align*}
	& R(A)-R(A^{g}) \\
	& <\Sigma _{j}\left( \frac{k-1}{4}N_{j}\limfunc{var}[\theta _{j}]+   \frac{1}{8}\limfunc{var}%
	[\theta _{ij}] +\frac{\sqrt{N_{j}}}{4}\sqrt{\limfunc{var}[\theta _{j}]}\sqrt{%
		\limfunc{var}[\theta _{ij}]}\right) \\
	& \quad  + \frac{1}{4} \limfunc{var} [\theta_{ij}]+ \Sigma_j \frac{3}{8}%
	\frac{N_{j} (\limfunc{var}[\theta _{ij}] + \limfunc{var}[\varepsilon_{ij}] )}{N_{j}-1} .
\end{align*}%
The dominant linear term is decreasing in $N_{j}$, and hence we know that as 
$N_{j}\rightarrow \infty $, revealing group identities is more
profitable.

\end{proof}

\begin{proof}[Proof of Proposition \protect\ref{vag3}]
Each consumer's demand function is given by 
\begin{equation*}
q_{i}=w_{i}-\left( \ell _{i}-x_{i}\right) ^{2}-p_{i}.
\end{equation*}%
This means the producer's profit is given by 
\begin{equation*}
\pi =\sum\nolimits_{i=1}^{N}p_{i}\left( w_{i}-\left( \ell _{i}-x_{i}\right)
^{2}-p_{i}\right) .
\end{equation*}%
Therefore, under any information structure $S$, the producer offers 
\begin{eqnarray*}
p_{i} &=&\left( \mathbb{E}\left[ w_{i}\mid S\right] -\mathbb{E}\left[ (\ell
_{i}-\mathbb{E}\left[ \ell _{i}\mid S\right] )^{2}\mid S\right] \right) /2,
\\
&=&\left( \mathbb{E}\left[ w_{i}\mid S\right] -\mathbb{E}\left[ (\ell _{i}-%
\mathbb{E}\left[ \ell _{i}\mid S\right] )^{2}\right] \right) /2, \\
x_{i} &=&\mathbb{E}\left[ \ell _{i}\mid S\right] ,
\end{eqnarray*}%
where the second line relies on the fact that the underlying random
variables are normal so that $\ell _{i}-\mathbb{E}[\ell _{i}|S]$ is
independent of $S$.

The consumer's surplus is then given by 
\begin{align*}
U_{i}\left( S\right) & =\frac{1}{2}\mathbb{E}\left[ \left( w_{i}-\left( \ell
_{i}-\mathbb{E}\left[ \ell _{i}\mid S\right] \right) ^{2}-\frac{\mathbb{E}%
\left[ w_{i}\mid S\right] -\mathbb{E}\left[ (\ell _{i}-\mathbb{E}\left[ \ell
_{i}\mid S\right] )^{2}\right] }{2}\right) ^{2}\right] \\
& =\frac{1}{2}\mathbb{E}\left[ \left( w_{i}-\frac{1}{2}\mathbb{E}\left[
w_{i}\mid S\right] \right) ^{2}\right] +\frac{1}{2}\mathbb{E}\left[ \left(
\left( \ell _{i}-\mathbb{E}\left[ \ell _{i}\mid S\right] \right) ^{2}-\frac{1%
}{2}\mathbb{E}\left[ \left( \ell _{i}-\mathbb{E}\left[ \ell _{i}\mid S\right]
\right) ^{2}\right] \right) ^{2}\right] \\
& -\mathbb{E}\left[ \left( w_{i}-\frac{1}{2}\mathbb{E}\left[ w_{i}\mid S%
\right] \right) \right] \mathbb{E}\left[ \left( \ell _{i}-\mathbb{E}\left[
\ell _{i}\mid S\right] \right) ^{2}-\frac{1}{2}\mathbb{E}\left[ \left( \ell
_{i}-\mathbb{E}\left[ \ell _{i}\mid S\right] \right) ^{2}\right] \right] \\
& =\frac{1}{2}\mathbb{E}\left[ w_{i}^{2}-\frac{3}{4}\mathbb{E}\left[
w_{i}\mid S\right] ^{2}\right] +\frac{1}{2}\mathbb{E}\left[ \left( \ell _{i}-%
\mathbb{E}\left[ \ell _{i}\mid S\right] \right) ^{4}-\frac{3}{4}\mathbb{E}%
\left[ \left( \ell _{i}-\mathbb{E}\left[ \ell _{i}\mid S\right] \right) ^{2}%
\right] ^{2}\right] \\
& -\frac{1}{4}\mathbb{E}\left[ w_{i}\right] \mathbb{E}\left[ \left( \ell
_{i}-\mathbb{E}\left[ \ell _{i}\mid S\right] \right) ^{2}\right] .
\end{align*}%
Therefore the difference is: 
\begin{eqnarray*}
U_{i}(S)-U_{i}(\varnothing ) &=&-\frac{3}{8}\limfunc{var}\left[ \mathbb{E}%
\left[ w_{i}\mid S\right] \right] +\frac{1}{4}\mu \limfunc{var}\left[ 
\mathbb{E}\left[ \ell _{i}\mid S\right] \right] +\frac{1}{2}\mathbb{E}\left[
\left( \ell _{i}-\mathbb{E}\left[ \ell _{i}\mid S\right] \right) ^{4}\right]
\\
&&-\frac{3}{8}\mathbb{E}\left[ \left( \ell _{i}-\mathbb{E}\left[ \ell
_{i}\mid S\right] \right) ^{2}\right] ^{2}-\frac{1}{2}\mathbb{E}\left[
\left( \ell _{i}-\mu _{\tau }\right) ^{4}\right] +\frac{3}{8}\mathbb{E}\left[
\left( \ell _{i}-\mu _{\tau }\right) ^{2}\right] ^{2}.
\end{eqnarray*}%
Because every random variable is assumed to be normal, $\ell _{i}-\mathbb{E}%
[\ell _{i}|S]$ is also normal with zero mean. We can further simplify and
obtain 
\begin{align*}
U_{i}(S)-U_{i}(\varnothing )& =-\frac{3}{8}\limfunc{var}\left[ \mathbb{E}%
\left[ w_{i}\mid S\right] \right] +\frac{1}{4}\mu \limfunc{var}\left[ 
\mathbb{E}\left[ \ell _{i}\mid S\right] \right] +\frac{3}{2}\left( \left[
\ell _{i}\right] -\limfunc{var}\left[ \mathbb{E}\left[ \ell _{i}\mid S\right]
\right] \right) ^{2} \\
& -\frac{3}{8}\left( \limfunc{var}\left[ \ell _{i}\right] -\limfunc{var}%
\left[ \mathbb{E}\left[ \ell _{i}\mid S\right] \right] \right) ^{2}-\frac{3}{%
2}\limfunc{var}[\ell _{i}]^{2}+\frac{3}{8}\limfunc{var}[\ell _{i}]^{2}, \\
=-\frac{3}{8}\limfunc{var}\left[ \mathbb{E}\left[ w_{i}\mid S\right] \right]
& +\frac{1}{4}\mu \limfunc{var}\left[ \mathbb{E}\left[ \ell _{i}\mid S\right]
\right] +\frac{9}{8}\left( \limfunc{var}\left[ \mathbb{E}\left[ \ell
_{i}\mid S\right] \right] ^{2}-2\limfunc{var}\left[ \mathbb{E}\left[ \ell
_{i}\mid S\right] \right] (\sigma _{\tau }^{2}+\sigma _{\tau
_{i}}^{2})\right) .
\end{align*}%
Similarly we have: 
\begin{align*}
\Pi _{i}(S)& =\frac{1}{4}\mathbb{E}\left[ \left( \mathbb{E}\left[ w_{i}\mid S%
\right] -\mathbb{E}\left[ (\ell _{i}-\mathbb{E}\left[ \ell _{i}\mid S\right]
)^{2}\right] \right) \right] \\
& =\frac{1}{4}\mathbb{E}\left[ \mathbb{E}\left[ w_{i}\mid S\right] ^{2}-2%
\mathbb{E}\left[ w_{i}\mid S\right] \mathbb{E}\left[ (\ell _{i}-\mathbb{E}%
\left[ \ell _{i}\mid S\right] )^{2}\right] +\mathbb{E}\left[ (\ell _{i}-%
\mathbb{E}\left[ \ell _{i}\mid S\right] )^{2}\right] ^{2}\right] ,
\end{align*}%
and hence 
\begin{eqnarray*}
\Pi _{i}(S)-\Pi _{i}(\varnothing ) &=&\frac{1}{4}\limfunc{var}[\mathbb{E}%
\left[ w_{i}\mid S\right] ]+\frac{1}{2}\mu \limfunc{var}\left[ \mathbb{E}%
\left[ \ell _{i}\mid S\right] \right] \\
&&+\frac{1}{4}\left( \limfunc{var}\left[ \mathbb{E}\left[ \ell _{i}\mid S%
\right] \right] ^{2}-2\limfunc{var}\left[ \mathbb{E}\left[ \ell _{i}\mid S%
\right] \right] (\sigma _{\tau }^{2}+\sigma _{\tau _{i}}^{2})\right) .
\end{eqnarray*}

To summarize, the impact of data sharing on social surplus is given by 
\begin{eqnarray*}
&&W_{i}(S)-W_{i}(\varnothing )=U_{i}(S)-U_{i}(\varnothing )+\Pi _{i}(S)-\Pi
_{i}(\varnothing ), \\
&=&-\frac{1}{8}\limfunc{var}\left[ \mathbb{E}\left[ w_{i}\mid S\right] %
\right] +\frac{3}{4}\mu \limfunc{var}\left[ \mathbb{E}\left[ \ell _{i}\mid S %
\right] \right] +\frac{11}{8}\left( \limfunc{var}\left[ \mathbb{E}\left[
\ell _{i}\mid S\right] \right] ^{2}-2\limfunc{var}\left[ \mathbb{E}\left[
\ell _{i}\mid S\right] \right] (\sigma _{\tau }^{2}+\sigma _{\tau
_{i}}^{2})\right) .
\end{eqnarray*}
Therefore the difference $W_{i}(S)-W_{i}(\varnothing )$ is a quadratic
function of the variance of the conditional expectation $x\triangleq 
\limfunc{var}[\mathbb{E}[\ell _{i}|S]]$. In particular, we let 
\begin{equation*}
g(x)\triangleq \frac{11}{8}x^{2}+\left( \frac{3}{4}\mu -\frac{11}{4}(\sigma
_{\tau }^{2}+\sigma _{\tau _{i}}^{2})\right) x.
\end{equation*}
As long as $3\mu >11(\sigma _{\tau }^{2}+\sigma _{\tau _{i}}^{2})$, this
function is positive and increasing in $x$, which means a higher $\limfunc{
var}[\mathbb{E}[\ell _{i}|S]]$ increases consumer surplus.

Finally, as in the proof of Proposition \ref{si}, aggregating $w_{i}$
increases $W_{i}(S)$ but keeps $\Pi (\varnothing )$ and $U_{i}(S_{-i})$
unchanged. Not aggregating $\ell _{i}$ increases $W_{i}(S)$ while keeping $%
\Pi (\varnothing )$ and $U_{i}(S_{-i})$ unchanged. Therefore it is optimal
for the intermediary to aggregate $w_{i}$ but not $\ell _{i}$.\bigskip
\end{proof}

\newpage

\section*{Appendix B\label{idd}}

In this Appendix, we explore the data intermediary's ability to offer
privacy guarantees in equilibrium by collecting less than perfect
information about the consumers' signals. Specifically, we consider the
additive data structure in (\ref{wa})-(\ref{ea}) and again assume that
fundamental and error terms have a joint Gaussian distribution. We then
specify a class of data policies that add common and idiosyncratic noise
terms $\xi $ and $\xi _{i}$ to the consumers' original (noisy) signals $%
s_{i} $. We then have the following data inflow,%
\begin{equation*}
x_{i}=\underset{=s_{i}}{\underbrace{w_{i}+\sigma e_{i}}}+\xi +\xi _{i}\text{,%
}
\end{equation*}%
and the intermediary chooses the variance of the additional noise terms ($%
\sigma _{\xi }^{2}$ and $\sigma _{\xi _{i}}^{2}$).

\begin{proposition}[Optimal Noise Structure]
\label{odi}\strut

\begin{enumerate}
\item The optimal data inflow policy adds no idiosyncratic noise; i.e., $%
\sigma _{\xi _{i}}^{\ast }=0$.

\item The optimal data inflow policy adds (weakly) positive aggregate noise;
i.e., $\sigma _{\xi }^{\ast }\geq 0$.
\end{enumerate}
\end{proposition}

\begin{proof}[Proof of Proposition \protect\ref{odi}]
Recall the formula in the proof of Proposition \ref{impact}, 
\begin{align*}
\Pi _{i}(Y_{i},Y)& =\frac{\limfunc{var}[\mathbb{E}[w_{i}|Y]]+\mu ^{2}}{4}, \\
U_{i}(Y_{i},Y)& =\frac{1}{2}\mathbb{E}[(\mathbb{E}[w_{i}|Y_{i}])^{2}-\frac{3%
}{4}(\mathbb{E}[w_{i}|Y])^{2}].
\end{align*}%
With a noisier report $X$, consumer $i$ will know $S_{i}$ and $X$ both on
path and off path. The producer will know $X$ on path and $X_{-i}$ if
consumer $i$ deviates. Thus the revenue of the intermediary is: 
\begin{align*}
\frac{R(X)}{N}& =\Pi _{i}((S_{i},X),X)-\Pi _{i}(S_{i},\varnothing
)+U((S_{i},X),X)-U((S_{i},X),X_{-i}), \\
& =\frac{\limfunc{var}[\mathbb{E}[w_{i}|X]]}{4}-\frac{3\limfunc{var}[\mathbb{%
E}[w_{i}|X]]}{8}+\frac{3\limfunc{var}[\mathbb{E}[w_{i}|X_{-i}]]}{8}, \\
& =-\frac{\limfunc{var}[\mathbb{E}[w_{i}|X]]}{8}+\frac{3\limfunc{var}[%
\mathbb{E}[w_{i}|X_{-i}]]}{8}.
\end{align*}%
Recall that 
\begin{equation*}
X_{i}=w_{i}+\sigma e_{i}+\xi +\xi _{i}=\theta +\theta _{i}+(\sigma
\varepsilon _{i}+\xi _{i})+(\sigma \varepsilon +\xi ).
\end{equation*}%
For ease of exposition, we rewrite $(\sigma \varepsilon _{i}+\xi _{i})$ as $%
\varepsilon _{i}$ and $(\sigma \varepsilon +\xi )$ as $\varepsilon $.
Because the intermediary can control the variance of $\xi ,\xi _{i}$ but not
the initial precision of the consumers' signals, we effectively have a lower
bound of the variance $\underline{\sigma }_{i}^{2}$ and $\underline{\sigma }%
^{2}$ on the new pair $\varepsilon _{i},\varepsilon $. Denote the variance
of $\theta $ as $\sigma _{\theta }^{2}$ and similarly for other variables.
It is straightforward to calculate that: 
\begin{align*}
\mathbb{E}[w_{i}|X]& =\frac{N\sigma _{\theta }^{2}+\sigma _{\theta _{i}}^{2}%
}{N^{2}(\sigma _{\theta }^{2}+\sigma _{\varepsilon }^{2})+N(\sigma
_{\varepsilon _{i}}^{2}+\sigma _{\theta _{i}}^{2})}\sum_{i^{\prime
}}x_{i^{\prime }}, \\
\mathbb{E}[w_{i}|X_{-i}]& =\frac{N\sigma _{\theta }^{2}}{(N-1)^{2}(\sigma
_{\theta }^{2}+\sigma _{\varepsilon }^{2})+(N-1)(\sigma _{\varepsilon
_{i}}^{2}+\sigma _{\theta _{i}}^{2})}\sum_{i^{\prime }\neq i}x_{i^{\prime }},
\\
R(X)& =\frac{3(N-1)N\sigma _{\theta }^{4}}{8((N-1)(\sigma _{\varepsilon
}^{2}+\sigma _{\theta }^{2})+\sigma _{\varepsilon _{i}}^{2}+\sigma _{\theta
_{i}}^{2})}-\frac{(N\sigma _{\theta }^{2}+\sigma _{\theta _{i}}^{2})^{2}}{%
8(N(\sigma _{\varepsilon }^{2}+\sigma _{\theta }^{2})+\sigma _{\varepsilon
_{i}}^{2}+\sigma _{\theta _{i}}^{2})}
\end{align*}

Now we are ready to prove the theorem. We argue it is optimal to set $\sigma
_{\varepsilon _{i}}^{2}=\underline{\sigma }_{i}^{2}$ (i.e., to set $\sigma
_{\xi _{i}}^{2}=0$ ). To show this result, suppose $\sigma _{\varepsilon
_{i}}^{2}>\underline{\sigma }_{i}^{2}$. Then there exists $\delta >0$ such
that augmenting the common noise to $\bar{\sigma}_{\varepsilon
}^{2}\triangleq \sigma _{\varepsilon }^{2}+\delta ^{2}$ and diminishing the
idiosyncratic noise to $\bar{\sigma}_{\varepsilon _{i}}^{2}\triangleq \sigma
_{\varepsilon _{i}}^{2}-(N-1)\delta ^{2}\geq \underline{\sigma }_{i}^{2},$
the profits of the intermediary will strictly increase. Too see this,
consider the expression of the revenue 
\begin{equation*}
R\left( S\right) =\frac{3(N-1)N\sigma _{\theta }^{4}}{8((N-1)(\sigma
_{\varepsilon }^{2}+\sigma _{\theta }^{2})+\sigma _{\varepsilon
_{i}}^{2}+\sigma _{\theta _{i}}^{2})}-\frac{(N\sigma _{\theta }^{2}+\sigma
_{\theta _{i}}^{2})^{2}}{8(N(\sigma _{\varepsilon }^{2}+\sigma _{\theta
}^{2})+\sigma _{\varepsilon _{i}}^{2}+\sigma _{\theta _{i}}^{2})}.
\end{equation*}%
The first term is unchanged under the new information structure, whereas the
denominator of the second term increases; thus, the total profit increases.
\end{proof}

As we establish in Proposition \ref{mis}, no profitable intermediation is
feasible for values of $\alpha $ less than a threshold that decreases with $%
N $. This threshold is independent of the correlation coefficient $\beta $
of the initial noise terms $\left( e_{i},e_{j}\right) $. Furthermore, as $%
\alpha $ approaches this threshold from above, the optimal level of common
noise grows without bound. Figure \ref{sigmaepsilon} shows the optimal
variance in the additional common noise term.

\begin{figure}
	\label{sigmaepsilon}
	\centering
	\caption{Optimal Additional Noise\ $(\sigma=1, N=2)$}%
	\includegraphics[width=0.6\textwidth]{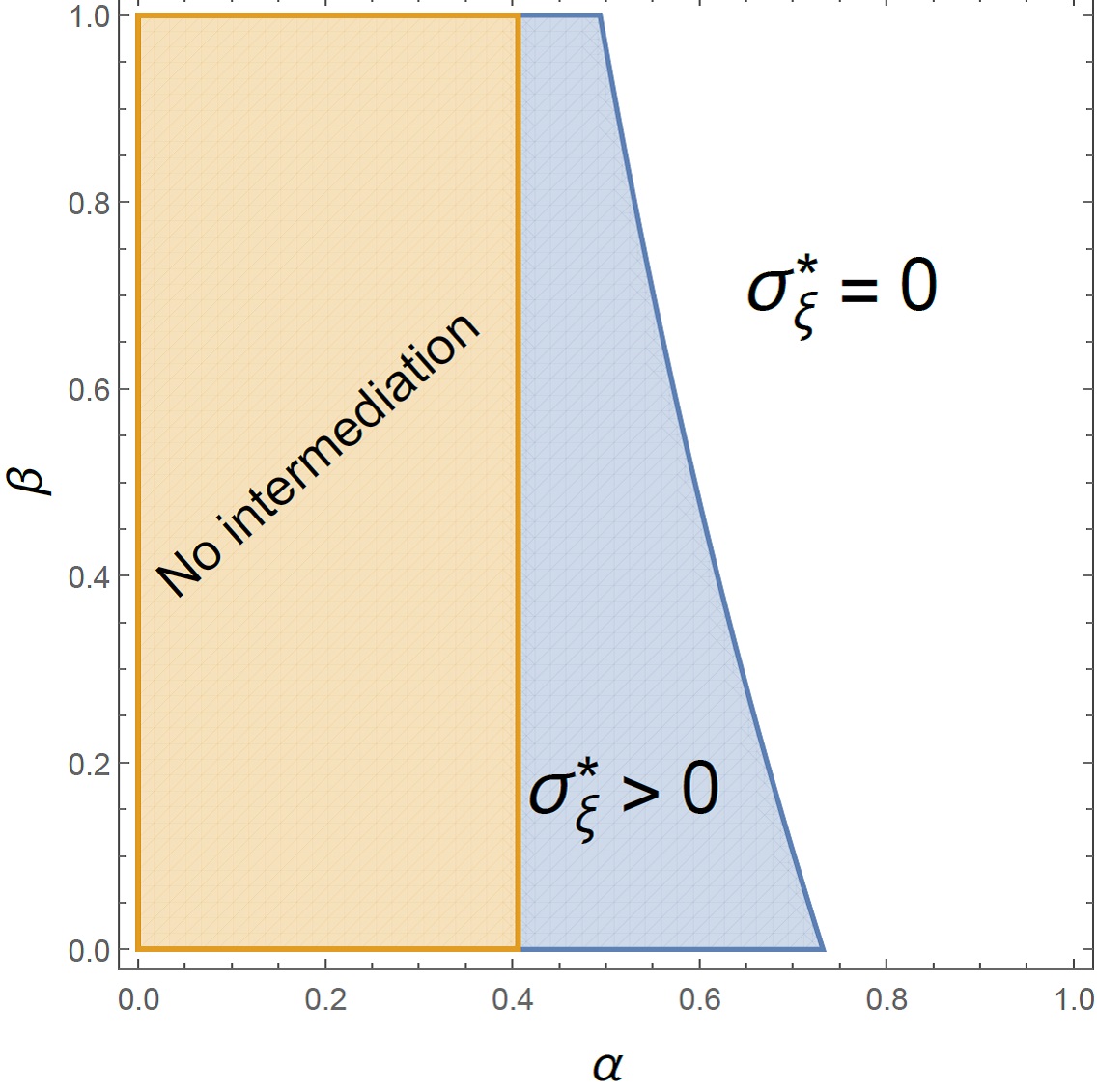}
\end{figure}

\begin{proposition}[Profitability of Data Intermediation]
\label{mis}\quad \newline
Under the optimal data policy, the intermediary's profits are strictly
positive if and only if 
\begin{equation*}
\alpha >\frac{N(\sqrt{3}+1)-1}{2N\left( N+1\right) -1}\in \left( 0,1\right) 
\text{.}
\end{equation*}
\end{proposition}

\begin{proof}[Proof of Proposition \protect\ref{mis}]
Recall that $\alpha $ is the correlation coefficient between $w_{i}$ and $%
w_{j}.$ Because we have normalized $\limfunc{var}\left[ w_{i}\right] =1$,
under the additive structure, we have $\sigma _{\theta }^{2}=\alpha $ and $%
\sigma _{\theta _{i}}^{2}=1-\alpha $. To establish the result in the
statement, we must then show that the intermediary obtains positive profits
if and only if 
\begin{equation*}
N\left( \sqrt{3}-1\right) \sigma _{\theta }^{2}-\sigma _{\theta _{i}}^{2}>0.
\end{equation*}

When $N\left( \sqrt{3}-1\right) \sigma _{\theta }^{2}<\sigma _{\theta
_{i}}^{2}$, we have: 
\begin{align*}
R& =\frac{3N\sigma _{\theta }^{4}}{8(\sigma _{\theta }^{2}+\sigma
_{\varepsilon }^{2})}\left( 1-\frac{\sigma _{\theta _{i}}^{2}+\sigma
_{\varepsilon _{i}}^{2}}{(N-1)(\sigma _{\varepsilon }^{2}+\sigma _{\theta
}^{2})+\sigma _{\varepsilon _{i}}^{2}+\sigma _{\theta _{i}}^{2}}\right) -%
\frac{(N\sigma _{\theta }^{2}+\sigma _{\theta _{i}}^{2})^{2}}{8N(\sigma
_{\theta }^{2}+\sigma _{\varepsilon }^{2})}\left( 1-\frac{\sigma _{\theta
_{i}}^{2}+\sigma _{\varepsilon _{i}}^{2}}{N(\sigma _{\varepsilon
}^{2}+\sigma _{\theta }^{2})+\sigma _{\varepsilon _{i}}^{2}+\sigma _{\theta
_{i}}^{2}}\right) \\
& \leq \left( \frac{3N\sigma _{\theta }^{4}}{8(\sigma _{\theta }^{2}+\sigma
_{\varepsilon }^{2})}-\frac{(N\sigma _{\theta }^{2}+\sigma _{\theta
_{i}}^{2})^{2}}{8N(\sigma _{\theta }^{2}+\sigma _{\varepsilon }^{2})}\right)
\left( 1-\frac{\sigma _{\theta _{i}}^{2}+\sigma _{\varepsilon _{i}}^{2}}{%
N(\sigma _{\varepsilon }^{2}+\sigma _{\theta }^{2})+\sigma _{\varepsilon
_{i}}^{2}+\sigma _{\theta _{i}}^{2}}\right) \leq 0.
\end{align*}%
Conversely, when $N\left( \sqrt{3}-1\right) \sigma _{\theta }^{2}>\sigma
_{\theta _{i}}^{2}$, we can rewrite $R$ as 
\begin{align*}
& \frac{A\sigma _{\varepsilon }^{2}+B}{8(N\sigma _{\varepsilon }^{2}+\sigma
_{\varepsilon _{i}}^{2}+N\sigma _{\theta }^{2}+\sigma _{\theta
_{i}}^{2})(N\sigma _{\varepsilon }^{2}-\sigma _{\varepsilon }^{2}+\sigma
_{\varepsilon _{i}}^{2}+N\sigma _{\theta }^{2}-\sigma _{\theta }^{2}+\sigma
_{\theta _{i}}^{2})}, \\
A& =(N-1)\left( 2N^{2}\sigma _{\theta }^{4}-2N\sigma _{\theta }^{2}\sigma
_{\theta _{i}}^{2}-\sigma _{\theta _{i}}^{4}\right) >0.
\end{align*}%
Therefore the intermediary can obtain a positive profit $R$ by setting $%
\sigma _{\varepsilon }^{2}$ sufficiently large through the addition of
correlated noise.
\end{proof}

The optimal level of common noise $\sigma _{\xi }^{\ast }$ is strictly
positive when the number of consumers $N$ or the correlation in their
willingness to pay $\alpha $ is sufficiently small: if the consumers'
preferences are sufficiently correlated, or if the market is sufficiently
large, the intermediary does not add any noise. If the consumers'
fundamentals are not sufficiently correlated, the intermediary makes their 
\emph{signals }more correlated with additional common noise $\sigma _{\xi
}^{\ast }.$

The additional noise reduces the amount of information procured from
consumers and hence the total compensation owed to them. These cost savings
come at the expense of lower revenues. In this respect, aggregation and
noise serve a common purpose. However, because the intermediary optimally
anonymizes the consumers' signals, the correlation in the supplemental noise
terms $\xi $ renders signal $s_{i}$ less valuable on the margin for
estimating the average willingness to pay $\bar{w}.$ In other words, the
aggregate demand information without consumer $i$'s signal $s_{i}$ is a
relatively better predictor of $\bar{w}$ when the intermediary uses common
rather than idiosyncratic noise. Therefore, by using common noise
exclusively, the intermediary can hold the information sold to the producer
constant while reducing the cost of acquiring that information from
consumers.

Finally, note that these two elements of information design---aggregation
and noise---interact richly with one another. In particular, the value of
common noise is deeply linked to that of aggregate data: if the intermediary
is restricted to complete (identity-revealing) data sharing, one can show
that supplemental idiosyncratic noise is optimal when the consumers' initial
signals $s_{i}$ are sufficiently precise.

\newpage

\bibliographystyle{econometrica}
\bibliography{ale2,general}

\end{document}